\documentclass[12pt, draftclsnofoot, onecolumn]{IEEEtran}
\usepackage[utf8]{inputenc}

\usepackage{verbatim}
\usepackage{color}
\usepackage{multicol}
\usepackage{placeins}
\usepackage{float}
\usepackage{graphicx}
\usepackage{cite}

\usepackage[caption=false,font=footnotesize]{subfig}

\usepackage{algorithmic}
\usepackage{algorithm}

\usepackage{array}
\newcolumntype{L}[1]{>{\raggedright\let\newline\\\arraybackslash\hspace{0pt}}m{#1}}
\newcolumntype{C}[1]{>{\centering\let\newline\\\arraybackslash\hspace{0pt}}m{#1}}
\usepackage{booktabs}

\usepackage{amsmath}
\usepackage{amssymb}
\usepackage{amsthm}
\usepackage{bm}

\newtheorem{proposition}{Proposition}

\DeclareMathOperator{\EX}{\mathbb{E}} 
\DeclareMathOperator{\Var}{\mathbb{V}} 
\newcommand{\transpose}{^{\mkern-1.5mu\mathsf{T}}}
\newcommand{\hermitian}{^{\mathsf{H}}} 
\newcommand{\conj}{^*} 

\DeclareMathOperator*{\argmin}{argmin}
\DeclareMathOperator*{\argmax}{argmax}

\DeclareMathOperator{\Tr}{Tr}
\DeclareMathOperator{\I}{\mathsf{j}}
\let\Re\relax
\DeclareMathOperator{\Re}{Re} 
\let\Im\relax
\DeclareMathOperator{\Im}{Im} 
\newcommand{\norm}[1]{\left\lVert #1 \right\rVert}
\newcommand{\set}[1]{\left\{#1\right\}}

\newcommand{\complexs}{\mathbb{C}}

\newcommand{\pilot}{\bm{\phi}}
\newcommand{\truepilot}{\bm{\phi}_n}
\newcommand{\falsepilot}{\bm{\phi}_{n'}}
\newcommand{\channel}{\bm{g}}
\newcommand{\truechannel}{\bm{g}_{n}}
\newcommand{\falsechannel}{\bm{g}_{n'}}
\newcommand{\ntruechannel}{\bm{z}_{n}}
\newcommand{\nfalsechannel}{\bm{z}_{n'}}
\newcommand{\uplinkchannel}{\bm{h}}
\newcommand{\noise}{\bm{W}}
\newcommand{\channelset}{\mathcal{G}}
\newcommand{\uplinkchannelset}{\mathcal{H}}
\newcommand{\Y}{\bm{Y}}
\newcommand{\SNR}{\rho}
\newcommand{\eSNR}{\rho M \beta}
\newcommand{\repart}[1]{\Re\left\{#1\right\}}
\newcommand{\impart}[1]{\Im\left\{#1\right\}}
\newcommand{\metric}[2]{\zeta_{\text{#2}}\left(#1\right)}
\newcommand{\abs}[1]{\left|#1\right|}
\newcommand{\pilotmatrix}{\bm{\Phi}}
\newcommand{\newpilotmatrix}{\bm{\Phi}^\text{new}}
\newcommand{\channelmatrix}{\bm{G}}
\newcommand{\CN}[2]{\mathcal{CN}\left( #1,#2 \right)}
\newcommand{\mostlikelypilot}{n: \pilot = s(\channel_n),\, \channel_n \in \channelset}
\newcommand{\mostlikelychannelandpilot}{\channel\in\channelset, \pilot:\pilot=s(\channel)}
\newcommand{\mostlikelychannel}{\channel\in\channelset}

\newcommand{\SequenceGeneratingAlgorithm}{SGA}
\newcommand{\UnknownAlgorithm}{UPSCA}
\newcommand{\KnownAlgorithm}{KPSCA}

\title{Combining Reciprocity and CSI Feedback in MIMO Systems}
\author{\IEEEauthorblockN{Ema~Becirovic,  Emil~Bj\"{o}rnson and Erik~G.~Larsson} \\
\thanks{This work was supported in part by Grants 2019-05068 and D0760701 from the Swedish Research Council (VR), in part by ELLIIT, and in part by the Knut och Alice Wallenberg (KAW) Foundation. Part of the material in this paper was presented at the 54th Asilomar Conference on Signals, Systems, and Computers, Pacific Grove (virtual), California, November 1--5, 2020 \cite{confver}.}
\thanks{E. Becirovic and E.~G.~Larsson are with the  Department of Electrical Engineering (ISY), Link\"{o}ping University,  Link\"{o}ping, Sweden (e-mail: \{ema.becirovic, erik.g.larsson\}@liu.se). E.~Bj\"ornson was with the  Department of Electrical Engineering (ISY), Link\"{o}ping University,  Link\"{o}ping, Sweden, and is now with the Dept. of Computer Science, KTH Royal Institute of Technology, Kista, Sweden (email: emilbjo@kth.se). }%
}

\begin{document}
\maketitle

\begin{abstract}
	Reciprocity-based time-division duplex (TDD) Massive MIMO (multiple-input multiple-output) systems utilize channel estimates obtained in the uplink to perform precoding in the downlink. However, this method has been criticized of breaking down, in the sense that the channel estimates are not good enough to spatially separate multiple user terminals, at low uplink reference signal signal-to-noise ratios, due to insufficient channel estimation quality. Instead, codebook-based downlink precoding has been advocated for as an alternative solution in order to bypass this problem. We analyze this problem by considering a ``grid-of-beams world'' with a finite number of possible downlink channel realizations. Assuming that the terminal accurately can detect the downlink channel, we show that in the case where reciprocity holds, carefully designing a mapping between the downlink channel and the uplink reference signals will perform better than both the conventional TDD Massive MIMO and frequency-division duplex (FDD) Massive MIMO approach. We derive elegant metrics for designing this mapping, and further, we propose algorithms that find good sequence mappings.

\end{abstract}
\IEEEpeerreviewmaketitle

\section{Introduction} 
Massive MIMO is the currently most compelling physical-layer access technique. 
Its information theory and fundamental limits are well understood \cite{redbook,massivemimobook}, and it was standardized and commercialized in 5G \cite{Boccardi14,Parkvall17}. 
Massive MIMO uses a large number of active antennas at the base station to adaptively direct the reception and transmission.
Since the transmissions are highly directive, spatial multiplexing of user terminals is possible. 
To properly steer the transmission, the base station needs to know the downlink channels.
Hence, channel estimation is very important in Massive MIMO.

When the downlink transmission is based on poor channel estimates, the base station can only loosely focus the signals at the users in the downlink. 
This has particularly large impact on the interference between users that are spatially multiplexed. 
When the channel estimates are good, beamforming methods such as zero-forcing can be used to spatially suppress interference between the users. 
This is no longer possible when the channel estimates are of low quality, thus there will be large interference between the users and in some cases it might even be preferable to serve one user per time-frequency resource instead of performing spatial multiplexing.

Massive MIMO can operate in both time-division duplex (TDD) and frequency-division duplex (FDD) mode. In these two modes, the techniques for channel estimation differ since in TDD mode reciprocity can be taken advantage of for the base station to learn the downlink channel while in FDD mode the terminal needs to first learn the downlink channel and then communicate it to the base station.

In FDD Massive MIMO where codebook-based precoding is used, the user terminal acquires an estimate of the downlink channel, e.g., the base station does a beam sweep and the terminal decides which of the beams fits best with the true channel, or through downlink pilots, and uses uplink signaling to transmit the (quantized) downlink channel response which corresponds to codebook indices \cite{Marzetta06,Love08,Caire10,Lee16,Huang18}. 
Note that codebook-based precoding is also used in some versions of TDD Massive MIMO, however, it is used analogously to FDD Massive MIMO in the sense that it does not utilize the reciprocity of TDD. Hence, we will refer to all codebook-based precoding as FDD \cite[Ch.~11]{Dahlman18}.
The precoder codebooks need to have high resolution, since the downlink channel estimates need to be of high quality to be able to suppress interference. 
Hence, the precoder codebook will grow very large when the number of antenna ports is large. 
This in turn requires significant uplink resources for feedback of precoder indices from the terminal to the base station; the number of feedback bits will grow linearly with the number of antenna ports \cite{Jindal06}. 
Furthermore, the number of precoder codebooks can be very large in Massive MIMO as there needs to exist at least one precoder codebook for every possible combination of number of base station antennas and number of transmit layers at each terminal. 
In order to combat these difficulties, works have been done which harness some reciprocal features, e.g. in the angular domain, in FDD systems \cite{Liu19, Abdallah20,Zhong20}.

Even though Massive MIMO can operate in both TDD and FDD mode, it is most compelling in TDD mode since uplink-downlink reciprocity of radio propagation can be exploited at the base station by estimating the downlink channel from uplink reference signals.
This makes TDD Massive MIMO scalable with respect to the number of base station antennas.
TDD is also the most commonly used duplex mode in 5G since most bands are reserved for TDD.
Conventionally, in reciprocity-based TDD Massive MIMO, the user terminals are assigned one reference signal from a predefined set of reference signals. 
The base station estimates the channels based on these reference signals. 
The uplink reference signals are designed to have low mutual correlation, so that users that transmit reference signals will not interfere with each other to any large extent, thus avoiding the phenomenon called pilot contamination \cite{Elijah16}.

One practical concern has been that of operation in conditions when the received reference signal signal-to-noise-ratio (SNR) is very low. 
The coherent gain (``useful received power'') in TDD Massive MIMO is proportional to $\beta\gamma$, where $\beta$ is the path loss and $\gamma$ is the channel estimate mean-square, which itself scales proportionally to $\beta$ at low SNR \cite[Ch.~3]{redbook}.
Hence, at low SNR, a 3-dB increase in path loss effectively yields a 6-dB loss in coherent beamforming gain. 

The base station can generally transmit with much higher power than the user terminals, which leads to a large difference between the SNR that can be achieved in the downlink and uplink. 
For example, if the downlink transmit power is 40~W while the uplink transmit power is 0.1~W, then the downlink SNR will be 26 dB higher than the uplink SNR. 
Consequently, the range (coverage) of the downlink reference signals used for downlink channel estimation is longer than the range of the uplink reference signals. 
Some users will have a high downlink SNR and a low uplink SNR.
This limits the use of reciprocity-based channel estimation techniques, where the base station learns the channel from uplink reference signals, to terminals that are in the center of the coverage area---where both the uplink and downlink SNRs are sufficiently high. 

When there is no structure in the channel and the minimal number of orthogonal resources are used for the purpose of channel estimation, the number of resources spent on downlink channel estimation scales with the number of base station antennas, while the uplink channel estimation scales with the number of (single antenna) terminals. By this strategy, which is used in practice, the effective SNR in the uplink scales with the number of users. However, to bridge the gap of more than 20~dB discrepancy between the uplink and downlink SNR, more than 100 users should be served in the same coherence interval. Hence, in the currently practical scenarios, the effective downlink SNR is still much higher than the effective uplink SNR.

An existing way to mitigate the problem of the uplink-downlink SNR discrepancy is to only assign a small fraction of the frequency band to cell-edge terminals, so that these can focus their power in that band and thereby increase the uplink SNR \cite[Ch.~6]{Asplund20}. 
However, this will greatly reduce the data rate that such a terminal can achieve, since the rate is proportional to the assigned bandwidth. 
Moreover, for the exemplified power numbers, the uplink SNR will remain smaller than the downlink SNR, unless only a fraction 1/400 of the bandwidth is assigned to the terminal.
Another solution is to only send the uplink reference signals on a small fraction of the frequency band and then extrapolate the channel estimates over the whole band. This solution works well when there is a structure in the channel, e.g. a line-of-sight channel. 

In summary, the problems with existing methods for the base station to obtain channel estimates for downlink precoding are:
\begin{itemize}
	\item With pure reciprocity-based precoding, relying on uplink pilots for channel estimation at the base station, coverage is a problem especially with large bandwidths.
	\item With pure codebook-based precoding, large amounts of uplink resources are required for terminal-to-base station feedback of quantized channel estimates. Here, low uplink SNRs is also a problem, however, it is not clear when codebook-based precoding and reciprocity based precoding break down. (In addition, downlink resources for downlink pilots are needed.)
\end{itemize}

It was claimed, e.g., in  \cite{Abdallah20},\cite[Ch.~13]{Asplund20} that in situations where the received reference signal SNR is insufficient, channel state information (CSI) feedback techniques, used in FDD Massive MIMO and in some cases TDD Massive MIMO, would be more effective than using the reciprocity. 
In this paper we explain why if reciprocity holds, it is always advantageous to exploit it. 

\subsection{Technical Contributions}
We consider a cell where the base station is equipped with multiple antenna ports and communicates with a user terminal. 
The number of antennas at the base station is not necessarily ``massive'', however, since channel estimation is especially important in Massive MIMO, this is the case we have studied.
We consider both the cases where the channel is reciprocal, operating in TDD mode with reciprocity calibrated arrays, and where the channel is non-reciprocal, operating in FDD mode or with an uncalibrated array. 
In the considered scenario, due to high downlink SNR, the terminal has access to an accurate estimate of the downlink channel. 

We propose a solution that uses the accurate downlink channel estimates, available at the terminal, and the reciprocity of the channel to improve on the state-of-the-art and circumvent the problem with low uplink SNR.
More precisely, when reciprocity holds, the proposed solution uses the downlink channel estimate, which is available at the user terminal, to redesign the uplink reference signal transmission to improve the channel estimation at the base station. Conventionally, the user transmits a known pre-determined reference sequence in the uplink, which the base station can utilize to estimate the channel that it was transmitted over. In the proposed solution, in contrast, the reference signal is not fixed, but instead selected by the user terminal as a function of the downlink channel estimate. The function mapping between the downlink channel estimate and the uplink reference signal is pre-determined and known, so that the base station can jointly estimate the channel and the reference signal. This reduces the estimation errors as compared to having a pre-determined (fixed) reference signal. We furthermore develop an algorithm that finds optimized mapping between the channel estimate and the uplink reference signal, based on pairwise error probability that we obtain in closed form.

We compare our proposed solution to the conventional FDD approach where the terminal feeds back an uplink reference signal based on the estimated downlink channel. We show through simulations that utilizing reciprocity through our proposed solution is always favorable when reciprocity holds since the channel estimates will be of higher quality.

In summary, the technical contributions of this paper are:
\begin{itemize}
	\item In the proposed solution, the terminal selects the uplink reference signal as a function of the downlink channel.
	\item We find elegant metrics that, when minimized, lead to minimal pairwise error probabilities. Hence, good sequence mappings minimize the metrics.
	\item We propose an algorithmic approach to find good sequence mappings.
	\item Numerical simulations show that, when reciprocity holds, the proposed solution outperforms both conventional reciprocity-based channel estimation, conventionally used in TDD Massive MIMO, and feeding back the quantized channel estimates, conventionally used in FDD Massive MIMO. 
\end{itemize}

Part of this work was presented in \cite{confver}. In that paper we only considered a subset of the cases that arise, see entries in Table~\ref{tab:cases} marked with $\dagger$ (to be explained in more detail later). More specifically, we only considered a subset of the cases where reciprocity holds. In this paper we complete the study by covering all the cases and also comparing against a non-reciprocal system.

\subsection{Notation}
We use italicized, $x$, bold lowercase, $\bm{x}$, bold uppercase, $\bm{X}$, and calligraphic, $\mathcal{X}$, characters  to denote scalars, vectors, matrices, and sets, respectively. Probability density functions are denoted by $p(\cdot)$. $\bm{I}$ is used to denote the identity matrix and $\bm{0}$ denotes an all-zero vector, both with appropriate (clear from the context) dimensions. $(\cdot)\conj$, $(\cdot)\transpose$, and $(\cdot)\hermitian$ denote the conjugate, the transpose, and the conjugate transpose, respectively. The imaginary unit is denoted by $\I$. $\mathcal{CN}(\bm{0},\bm{R})$ denotes the (multivariate) circularly symmetric complex Gaussian distribution with covariance matrix $\bm{R}$. All the norms, $\norm{\cdot}$, in this paper denote the Frobenius norm. We denote Hadamard division by $\oslash$.

\section{System Model with Grid-of-Beams}\label{sec:system-model}

We study the case when a single single-antenna terminal is served by a base station equipped with $M$ antennas. We consider a grid-of-beams world; the downlink channel, $\channel \in \complexs^M$, is unknown at the base station but is one of the $N$ vectors in the set \begin{equation}
\channelset = \set{\channel_1,\dots,\channel_N},
\end{equation} where $N$ is in practice often larger than $M$. For instance, such a world can be useful to study when grid-of-beams-based algorithms are used for channel estimation, i.e., beam sweeping type algorithms. We can also view the grid-of-beams world as a special case of the framework in \cite{Hoon12,Adhikary13}. In \cite{Hoon12,Adhikary13}, the world is split into clusters where the users within each such bin have the same statistics. One way to interpret $\channelset$ is as a set of ``typical''  channels from each bin. Further, due to asynchronicity between the terminal and the base station and/or uncertainties of the distance between them, the channel might be phase shifted in which case the effective downlink channel is $e^{\I\theta}\channel$, where $\theta$ is the phase shift. This can happen if e.g., there is a phase mismatch in the sender and receiver chains of the terminal.

The uplink channel is $\uplinkchannel \in \mathbb{C}^M$. Reciprocity may hold, in which case $\uplinkchannel = e^{\I\theta}\channel$. In practice, there
will be a mismatch between the true channel and the ones
in $\channelset$, but we disregard this in the analytical part to derive methods of detecting said channel, i.e., we assume perfect calibration. 
This system model is also applicable in ``continuous'' worlds where the true channel is instead quantized to one of the beams in $\channelset$, as shown in the numerical results in Section~\ref{sec:results}.
Further, the base station array can be calibrated such that the uplink channel is one of the vectors in the set $\uplinkchannelset\neq\channelset$, $\uplinkchannel \in \uplinkchannelset$.
All the vectors in the set of downlink channels have the same norm, 
\begin{equation}
\norm{\bm{g}}^2 = M\beta,
\end{equation}
that is known at the base station. 

The goal of the base station is to estimate (detect) the downlink channel based on a signal, $\pilot\in\complexs^\tau$, where we assume that the signal is $\tau$ samples long, transmitted in the uplink. The terminal sends this signal and the base station receives
\begin{equation}
\Y = \sqrt{\SNR}\uplinkchannel\pilot\transpose + \noise  \in\complexs^{M\times\tau}, \label{eq:recievedsignal}
\end{equation}
where $\noise\in\complexs^{M\times\tau}$ is noise with independent $\CN{0}{1}$ elements. The sequence, $\pilot$ is either deterministic or chosen as a function of $\channel$. In either case, it has unit norm
\begin{equation}
\norm{\pilot}^2 = 1,
\end{equation}
which makes the parameter $\SNR$ have the interpretation of SNR after coherent integration over the sequence. For convenience, we define $\channelmatrix = [\channel_1,\dots,\channel_N]\in\complexs^{M\times N}$. 

\section{Downlink Channel Detection}\label{sec:dl-channel-detection}
In this section we will present how the base station detects the downlink channel. The detection process will be different depending on the uplink channel, i.e., if reciprocity holds or not, and whether the terminal knows the downlink channel, either partially or fully, or not. 
There are eight cases of concern, see Table~\ref{tab:cases}.  First of all, we do not study the cases where reciprocity does not hold and the terminal does not know the downlink channel $\channel$ since there is no way for the base station to gain knowledge about the downlink channel in this case. Additionally, the case where the terminal only knows $\theta$ is not treated since this situation will not occur in practice. However, note that the result in this scenario is a special case of when the terminal has full CSI, i.e., knows both $\channel$ and $\theta$, treated in Section~\ref{sec:reciprocity-full-csi}. Similarly, the case when the terminal has no CSI, i.e., neither knows $\channel$ nor $\theta$, treated in Section~\ref{sec:reciprocity-no-csi}, is a special case of the scenario where $\channel$ is known and $\theta$ is unknown, treated in Section~\ref{sec:reciprocity-only-g}. Nonetheless, the case with no CSI is treated explicitly since it corresponds to the case where the terminal has not acquired channel knowledge, which using a fixed uplink pilot corresponds to the conventional reciprocity-based channel estimation schemes. 

The cost of estimating the downlink channel at the terminal is not considered in this paper. In practice, if there is no structure of the channel to be exploited, the resources needed to estimate the channel reliably scales with $M$, which can be very large in massive MIMO \cite[Ch.~1]{redbook}\cite[Ch.~1]{massivemimobook}. However, if there is such structure, e.g., it is known that $\channel \in \channelset$, the required resources are much fewer \cite{Bjornson16}. Nonetheless, the terminal will act as if the downlink channel estimate it has is the true channel.

\begin{table}
	\caption{The cases of downlink channel detection covered in the paper. Preliminary results of cases marked with $\dagger$ are discussed in \cite{confver}.}\label{tab:cases}
	\centering
	\begin{tabular}{|L{2cm}|C{3cm}|C{3cm}|C{3cm}|C{3cm}|}
		\hline
		& \multicolumn{4}{c|}{terminal knowledge}\\ \cline{2-5}
		& \multicolumn{2}{c|}{$\channel$ unknown} & \multicolumn{2}{c|}{$\channel$ known} \\ \cline{2-5}
		& $\theta$ unknown & $\theta$ known & $\theta$ unknown & $\theta$ known\\ \hline
		reciprocity holds & No CSI at the terminal, \newline Section~\ref{sec:reciprocity-no-csi}$^\dagger$ & {
			Does not occur in practice, special case of Section~\ref{sec:reciprocity-full-csi}} & Partial CSI at the terminal, \newline Section~\ref{sec:reciprocity-only-g} & Full CSI at the terminal, \newline Section~\ref{sec:reciprocity-full-csi}$^\dagger$ \\ [7pt]
		\hline 
		reciprocity does not hold  & \multicolumn{2}{c|}{Not of interest -- beamforming cannot be performed}  & \multicolumn{2}{c|}{Section~\ref{sec:no-reciprocity}} 
		\\ [7pt] \hline
		reciprocity does not hold, but $\uplinkchannel \in \uplinkchannelset$ &  \multicolumn{2}{c|}{Not of interest -- beamforming cannot be performed}  & \multicolumn{2}{c|}{Section~\ref{sec:calibrated-array}}
		\\ \hline
	\end{tabular}
\end{table}

\subsection{Reciprocity Holds}\label{sec:reciprocity}
When reciprocity holds, the uplink channel is the same as the effective downlink channel, $\uplinkchannel = e^{\I\theta}\channel$.
The received signal \eqref{eq:recievedsignal} at the base station is in this case
\begin{equation}
	\Y = \sqrt{\SNR}e^{\I\theta}\channel\pilot\transpose+\noise.\label{eq:reciprocity-recieved-signal}
\end{equation}
If the terminal knows the phase shift $\theta$, we assume that it will compensate for it by rotating the symbols of $\pilot$ such that the base station receives
\begin{equation}
\Y = \sqrt{\SNR}\channel\pilot\transpose+\noise.\label{eq:reciprocity-recieved-signal-theta-known}
\end{equation}
There are three cases of channel estimation in this case, which we will consider below.

\subsubsection{Neither $\channel$ nor $\theta$ are known}\label{sec:reciprocity-no-csi}
In the first case, the terminal has no CSI which means that both the phase shift and the channel are unknown at the terminal.
Both the base station and the terminal know the sequence, which effectively becomes a pilot. Because the
terminal does not know the channel, the same pilot will be
transmitted regardless of which of the beams in $\channelset$ the channel corresponds to. This
setup resembles standard pilot-based channel estimation but
with the additional side-information that the channel is one of
the vectors in $\channelset$, and affected by a phase shift. 

First, the base station receives the signal \eqref{eq:reciprocity-recieved-signal}. 
The base station jointly finds the beam and the phase shift with the highest likelihood value, i.e., the maximum-likelihood estimate, through
\begin{equation}
\argmax_{\theta,\channel\in\channelset} p\left(\Y;\channel,\theta\right) = \argmin_{\theta,\channel\in\channelset} \norm{\Y-\sqrt{\SNR}e^{\I\theta}\channel\pilot\transpose},
\end{equation}
by first determining the maximum-likelihood estimate of the phase shift conditioned on each beam in $\channelset$ \cite[Ch.~6]{kay98detection}, 
\begin{align}
\hat{\theta}_n &= \argmax_{\theta} p\left(\Y;\channel_n,\theta\right)
= -\arg \pilot\transpose\Y\hermitian\channel_n,\label{eq:ML-estimate-phase-shift}
\end{align}
for $n=1,\dots,N$.
The base station continues by using the phase shift estimates in the detector which finds the beam with largest concentrated likelihood value, i.e., using the maximum likelihood detector	 \cite[Ch.~6]{kay98detection}
\begin{align}
\argmax_{\mostlikelychannel} p\left(\Y;\channel,\hat{\theta}\right) =\argmin_{\mostlikelychannel}  \norm{\Y - \sqrt{\SNR}e^{\I\hat{\theta}}\channel\pilot\transpose} ,
\end{align}
which is equivalent to
\begin{align}
\max_{\mostlikelychannel} \repart{ e^{\I\hat{\theta}}\pilot\transpose\Y\hermitian \channel} &= \max_{\mostlikelychannel} \repart{ e^{-\I\arg\pilot\transpose\Y\hermitian\channel}\pilot\transpose\Y\hermitian \channel} \nonumber\\
&= \max_{\mostlikelychannel} \abs{\pilot\transpose\Y\hermitian\channel}.
\end{align}

If the true channel is $\truechannel$, we make an error in favor of another channel, $\falsechannel$, if 
\begin{align}
\abs{\pilot\transpose\Y\hermitian\falsechannel} &>\abs{\pilot\transpose\Y\hermitian\truechannel} \iff\\
\abs{\sqrt{\SNR}\truechannel\hermitian\falsechannel+e^{\I\theta}\pilot\transpose\noise\hermitian\falsechannel} &>\abs{\sqrt{\SNR}M\beta+e^{\I\theta}\pilot\transpose\noise\hermitian\truechannel}.
\end{align}
\begin{proposition}\label{prop:pep-theta-unknown}
	In the case where reciprocity holds, and the terminal neither knows the channel, $\channel$, nor the phase shift, $\theta$, the pairwise error probability, the probability that channel $\falsechannel$ will be detected given that the true channel is $\truechannel$, will decrease when $\abs{\truechannel\hermitian\falsechannel}$ decreases or $\sqrt{\SNR M \beta}$ increases.
\end{proposition}
\begin{proof}
	See Appendix~\ref{app:proof-pep-theta-unknown}.
\end{proof}
Since the same pilot is used regardless of which beam in $\channelset$ the channel corresponds to, and the pilot has unit norm, the choice of pilot will not affect the performance of the detector. That is, if we are given a grid-of-beams world and the terminal does not know, or cannot accurately detect the beam, the only thing we can do to improve the detection performance is to either increase the power, or increase the length of the sequence $\pilot$, effectively increasing the SNR. This is true, regardless of the set of beams, $\channelset$.

\subsubsection{Only $\channel$ is known, not $\theta$}\label{sec:reciprocity-only-g} 
In this case, the terminal has accurately detected the downlink beam $\channel$ from e.g., downlink pilots, but lacks the knowledge of the phase shift $\theta$. This scenario can occur e.g., when there is a phase mismatch in the receiver and transmitter chains of the terminal. In this case, the terminal decides the uplink sequence based on the downlink channel through a mapping
\begin{equation}
\pilot = s(\channel). \label{eq:seqmap}
\end{equation}
In the special case when the terminal maps all possible channels to the same sequence, we fall back to the scenario in Section~\ref{sec:reciprocity-no-csi}.
There are $N$ (possibly non-unique) sequences. The sequences are paired with their corresponding beam such that the $n$:th sequence is based on the $n$:th beam in $\channelset$, $\truepilot = s(\truechannel)$ for $n = 1,\dots,N$. The sequences are for convenience grouped in a matrix $\pilotmatrix = [\pilot_1,\dots,\pilot_N]\in\complexs^{\tau\times N}$. 
The base station uses the phase shift estimate \eqref{eq:ML-estimate-phase-shift}, with $\pilot$ replaced with $\truepilot$, and finds the channel-sequence \emph{pair} with the largest concentrated likelihood value from the received signal \eqref{eq:reciprocity-recieved-signal} through the maximum-likelihood detector \cite[Ch.~6]{kay98detection}
\begin{align}
	\argmax_{\mostlikelychannelandpilot} p\left(\Y;\channel, \pilot, \hat{\theta} \right) &= \argmin_{\mostlikelychannelandpilot} \norm{\Y-\sqrt{\SNR}e^{\I\hat{\theta}}\channel\pilot\transpose}\\
	&=\argmax_{\mostlikelychannelandpilot} \abs{\pilot\transpose\Y\hermitian\channel}.
\end{align}
If the true channel-sequence pair is $(\truechannel,\truepilot)$, we make an error in favor of another channel-sequence pair, $(\falsechannel,\falsepilot)$, if 
\begin{align}
\abs{\falsepilot\transpose\Y\hermitian\falsechannel} &>\abs{\truepilot\transpose\Y\hermitian\truechannel} \iff\\
\abs{\sqrt{\SNR}\truepilot\hermitian\falsepilot\truechannel\hermitian\falsechannel+e^{\I\theta}\pilot\transpose\noise\hermitian\falsechannel} &>\abs{\sqrt{\SNR}M\beta+e^{\I\theta}\pilot\transpose\noise\hermitian\truechannel}.
\end{align}
\begin{proposition}\label{prop:pep-partial-csi}
	In the case where reciprocity holds, the terminal knows the channel, $\channel$, but not the phase shift, $\theta$, and the terminal selects the uplink reference sequence based on $\channel$, the pairwise error probability, the probability that channel $\falsechannel$ will be detected given that the true channel is $\truechannel$, will decrease when $\abs{\truepilot\hermitian\falsepilot\truechannel\hermitian\falsechannel}$ decreases or when $\sqrt{\SNR M \beta}$ increases.
\end{proposition}
\begin{proof}
	See Appendix~\ref{app:proof-pep-theta-unknown}. 
\end{proof}
Therefore, choosing the mapping $\pilot = s(\channel)$ is of great importance here. From the pairwise error probability we have an elegant metric to evaluate our mapping. At high SNRs, the pairwise error probability for the ``worst'' pair will dominate. Thus, we should find the channel-sequence mapping such that 
\begin{equation}
	\min_{s:\pilot=s(\channel)} \max_{\substack{n,n' \\ n \neq n'}} \abs{\truepilot\hermitian\falsepilot\truechannel\hermitian\falsechannel}  = \min_{s:\pilot=s(\channel)} \metric{\channelmatrix,\pilotmatrix}{U},
\end{equation}
where $\metric{\channelmatrix,\pilotmatrix}{U} = \displaystyle\max_{\substack{n,n' \\ n \neq n'}} \abs{\truepilot\hermitian\falsepilot\truechannel\hermitian\falsechannel}$. In Section~\ref{sec:reciprocity-channel-sequence-algorithm}, we will develop an algorithm that aims to minimize $\metric{\channelmatrix,\pilotmatrix}{U}$.

\subsubsection{Both $\channel$ and $\theta$ are known}\label{sec:reciprocity-full-csi}
In the last case where reciprocity holds, the terminal has full CSI, meaning that it knows both the beam $\channel$ and the phase shift $\theta$, which it acquired from e.g. downlink pilots.

Here, the base station detects the most likely beam from the received signal \eqref{eq:reciprocity-recieved-signal-theta-known}. The maximum likelihood detector is 
\cite[Ch.~6]{kay98detection}
\begin{align}
\argmax_{\mostlikelychannelandpilot} p\left(\Y;\channel\right) &=\argmin_{\mostlikelychannelandpilot}  \norm{\Y - \sqrt{\SNR}\channel\pilot\transpose}\\
&=\argmax_{\mostlikelychannelandpilot}\repart{\pilot\transpose\Y\hermitian\channel}.
\end{align}

If the true channel is $\truechannel$, we make an error in favor of another channel, $\falsechannel$, if 
\begin{align}
\repart{\falsepilot\transpose\Y\hermitian\falsechannel} &> \repart{\truepilot\transpose\Y\hermitian\truechannel} \iff\\
\repart{\sqrt{\SNR}\truepilot\hermitian\falsepilot\truechannel\hermitian\falsechannel+\falsepilot\transpose\noise\hermitian\falsechannel} &> \repart{\sqrt{\SNR}M\beta + \truepilot\transpose\noise\hermitian\truechannel}.
\end{align}
\begin{proposition}\label{prop:pep-full-csi}
	When reciprocity holds, and the terminal knows both the channel, $\channel$, and the phase shift, $\theta$, and the terminal selects the uplink reference sequence based on $\channel$, the pairwise error probability, the probability that channel $\falsechannel$ will be detected given that the true channel is $\truechannel$, is
	\begin{equation}
	\Pr\left( \textnormal{beam }n'\textnormal{ detected} \mid \textnormal{beam }n\textnormal{ correct} \right) = Q\left(\sqrt{\SNR (M\beta-\repart{\truepilot\hermitian\falsepilot\truechannel\hermitian\falsechannel})}\right),
	\end{equation}
	where $Q(\cdot)$ denotes the Q-function.
\end{proposition}
\begin{proof}
	See Appendix~\ref{app:proof-pep-theta-known}.
\end{proof}
Here, the pairwise error probability will decrease when $\repart{\truepilot\hermitian\falsepilot\truechannel\hermitian\falsechannel}$ decreases or when $\sqrt{\SNR M \beta}$ increases. Again, as in Section~\ref{sec:reciprocity-only-g}, choosing the mapping $\pilot = s(\channel)$ is of great importance here. We have a similar metric as before to evaluate the mapping. We should find the channel-sequence mapping such that 
\begin{equation}
\min_{s:\pilot=s(\channel)} \max_{\substack{n,n' \\ n \neq n'}} \repart{\truepilot\hermitian\falsepilot\truechannel\hermitian\falsechannel}  = \min_{s:\pilot=s(\channel)} \metric{\channelmatrix,\pilotmatrix}{K},
\end{equation}
where $\metric{\channelmatrix,\pilotmatrix}{K} = \displaystyle\max_{\substack{n,n' \\ n \neq n'}} \repart{\truepilot\hermitian\falsepilot\truechannel\hermitian\falsechannel}$. In Section~\ref{sec:reciprocity-channel-sequence-algorithm}, we will develop an algorithm which aims to minimize $\metric{\channelmatrix,\pilotmatrix}{K}$.

\subsection{Reciprocity Does not Hold}\label{sec:no-reciprocity}
In this case, reciprocity does \emph{not} hold but the terminal knows the downlink channel $\channel$ perfectly. The terminal selects the uplink sequence as a function of the downlink channel and the base station receives $\Y$ as in $\eqref{eq:recievedsignal}$, and detects  $\channel \in \channelset$ with the largest likelihood value. In this case, the knowledge of $\theta$ will not affect the detection problem at the base station, however, it will affect the downlink data rate since the transmission would need to use additional demodulation pilots to find the phase shift or transmit with non-coherent techniques.

Due to the lack of reciprocity, the uplink channel $\uplinkchannel$ is random and independent of the downlink channel $\channel$. Assuming that $\uplinkchannel$ has probability density function $p(\uplinkchannel)$, which is known at the base station, the detector detects the sequence with the largest likelihood value, which corresponds to the downlink channel in $\channelset$, from the received signal
\begin{align}
	\argmax_{\mostlikelypilot} p\left(\Y|\pilot\right) &= \argmax_{\mostlikelypilot} 	\int d\uplinkchannel\; p(\uplinkchannel)p\left( \Y | \pilot,\uplinkchannel \right). \label{eq:marginalized-likelihood-no-rec}
\end{align}
Here, we assume that we have a structure in the downlink channel corresponding to the finite set $\channelset$. However, the structure (if it exists) in the uplink channel is not fully known. Hence, we assume that the uplink channel is continuous. We study \eqref{eq:marginalized-likelihood-no-rec} in three cases: 
\begin{enumerate}
	\item The uplink channel is independent and identically distributed (i.i.d.) Rayleigh fading, $\uplinkchannel = \sqrt{\beta}\bm{z}$ where $\bm{z} \sim \CN{\bm{0}}{\bm{I}}$. The conditional probability in \eqref{eq:marginalized-likelihood-no-rec} is written as
	\begin{align}
		p\left( \Y|\pilot \right) &= \int d\uplinkchannel\; p(\uplinkchannel)p\left( \Y | \pilot,\uplinkchannel \right) \\
		&= \int d\bm{z}\; p(\bm{z})p\left( \Y | \pilot,\bm{z} \right) \\
		&= \frac{1}{(2\pi)^M}\frac{1}{(2\pi)^{M\tau}} \int d\bm{z}\; e^{-\norm{\bm{z}}^2}e^{-\norm{\Y-\sqrt{\SNR\beta}\bm{z}\pilot\transpose}^2}. \label{eq:sum-exponents-rayleigh}
	\end{align}
	The sum of the two negative exponents in \eqref{eq:sum-exponents-rayleigh} can be written 
	\begin{align}
		&\norm{\bm{z}}^2 + \norm{\Y}^2 + \SNR\beta\norm{\bm{z}}^2-2\sqrt{\SNR\beta}\repart{\pilot\transpose\Y\hermitian\bm{z}} \\
		&=\norm{\sqrt{\SNR\beta+1}\bm{z}-\sqrt{\frac{\SNR\beta}{\SNR\beta+1}}\Y\pilot\conj}^2+\norm{\Y}^2-\frac{\SNR\beta}{\SNR\beta+1}\norm{\Y\pilot\conj}^2. \label{eq:sum-exponents-rayleigh-2}
	\end{align}
	The exponential of the first term in \eqref{eq:sum-exponents-rayleigh-2} integrates to a constant and the second term is independent of the sequence $\pilot$. Hence, maximization of $p(\Y|\pilot)$ is equivalent to maximization of $\norm{\Y\pilot\conj}^2$,
	\begin{equation}
		\argmax_{\mostlikelypilot} p\left(\Y|\pilot\right) = \argmax_{\mostlikelypilot} \norm{\Y\pilot\conj}^2. \label{eq:no-reciprocity-rayleigh}
	\end{equation}
	If the true sequence is $\truepilot$, we make an error in favor of another sequence, $\falsepilot$, if
	\begin{align}
		\norm{\Y\falsepilot\conj} &> \norm{\Y\truepilot\conj} \iff \\
		\norm{\sqrt{\SNR}\uplinkchannel \truepilot\transpose\falsepilot\conj + \noise\falsepilot\conj}^2 &> \norm{\sqrt{\SNR}\uplinkchannel  + \noise\truepilot\conj}^2.
	\end{align}
	We let $\alpha = \truepilot\hermitian\falsepilot$ and $\bm{U}$ be a unitary matrix where the first two rows are the result of Gram-Schmidt orthogonalization of $\truepilot$ and $\falsepilot$. The first row is $\bm{u}_1\transpose = \truepilot\transpose$ and the second row is $\bm{u}_2\transpose = \frac{\falsepilot\transpose-\alpha\truepilot\transpose}{\sqrt{1-\abs{\alpha}^2}}$. The other rows are orthogonal to both $\truepilot\transpose$ and $\falsepilot\transpose$. Since the elements in $\noise$ are i.i.d. $\CN{0}{1}$, the elements of $\noise \bm{U}$ are also i.i.d. $\CN{0}{1}$. Therefore we can write the pairwise error probability as	
	\begin{align}
	&\Pr\left( \norm{\sqrt{\SNR}\uplinkchannel\truepilot\transpose\falsepilot\conj + \noise\falsepilot\conj}^2 > \norm{\sqrt{\SNR}\uplinkchannel +\noise\truepilot\conj}^2 \right) \\
	&=\Pr\left( \norm{\sqrt{\SNR}\alpha\conj\uplinkchannel + \noise\bm{U}\falsepilot\conj}^2 > \norm{\sqrt{\SNR}\uplinkchannel +\noise\bm{U}\truepilot\conj}^2 \right) \\
	&= \Pr\left( \norm{\sqrt{\SNR}\alpha\conj\uplinkchannel + \noise\left( \alpha\conj\bm{e}_1 + \sqrt{1-\abs{\alpha}^2}\bm{e}_2 \right)}^2 > \norm{\sqrt{\SNR}\uplinkchannel +\noise\bm{e}_1}^2 \right) \\
	&\overset{(a)}{=} \Pr\left( \norm{ \abs{\alpha}\bm{x} + \sqrt{1-\abs{\alpha}^2}\bm{y} }^2 > \norm{\bm{x}}^2  \right) \label{eq:no-rec-pep}
	\end{align}
	where $\bm{e}_1$ and $\bm{e}_2$ are the first and second column of $\bm{I}_{\tau}$, respectively, and in (a) we introduced $\bm{x} = \sqrt{\SNR}\uplinkchannel+\noise\bm{e}_1 \sim \CN{\bm{0}}{(\SNR\beta+1)\bm{I}_M}$ and $\bm{y} = \noise\bm{e}_2 \sim \CN{\bm{0}}{\bm{I}_M}$. Note that $\bm{x}$ and $\bm{y}$ are independent. 
	The pairwise error probability \eqref{eq:no-rec-pep} will increase when $\abs{\alpha}$ is increasing, see Appendix~\ref{app:proof3} for a proof.
	Therefore, here we want to use sequences that do not have high correlation.
	In this case, the optimal sequence mapping is independent of the set of downlink channels, $\channelset$.
	The set of sequences that is optimal fulfills
	\begin{equation}
		\min_{s:\pilot=s(\channel)} \max_{\substack{n,n' \\ n \neq n'}} \abs{\truepilot\hermitian\falsepilot} = \min_{s:\pilot=s(\channel)} \metric{\pilotmatrix}{NR},
	\end{equation}
	where $\metric{\pilotmatrix}{NR} = \displaystyle \max_{\substack{n,n' \\ n \neq n'}} \abs{\truepilot\hermitian\falsepilot}$.
	\item The uplink channel is a uniform linear array with $\lambda/2$-spacing line-of-sight channel, where $\lambda$ is the wavelength, $\uplinkchannel =  \sqrt{\beta}e^{\I\xi}\,\left[\begin{smallmatrix}
	e^0& e^{-\I\pi\sin\psi}& \dots& e^{-\I\pi(M-1)\sin\psi} \end{smallmatrix}\right]\transpose$, which is parameterized by $\psi$, the impinging angle, $p(\psi) = \frac{1}{\pi}, \frac{-\pi}{2} < \psi \leq \frac{\pi}{2}$, and $\xi$, the phase shift stemming from the unknown distance to the terminal $p(\xi) = \frac{1}{2\pi}, -\pi < \xi \leq \pi.$
	The sequence with the largest likelihood value is
	\begin{align}
	\argmax_{\mostlikelypilot} p(\Y|\pilot) &= \argmax_{\mostlikelypilot} \int_{-\pi}^{\pi}d\xi\,\int_{-\frac{\pi}{2}}^{\frac{\pi}{2}}d\psi\, p(\Y|\pilot,\psi,\xi)p(\psi)p(\xi) \\
	&= \argmax_{\mostlikelypilot}  \int_{-\pi}^{\pi}d\xi\,\int_{-\frac{\pi}{2}}^{\frac{\pi}{2}}d\psi\, \frac{1}{2\pi}\frac{1}{\pi}\frac{1}{\pi^{M\tau}}e^{-\norm{\Y-\sqrt{\SNR}\uplinkchannel\pilot\transpose}^2} \\
	&= \argmax_{\mostlikelypilot} \int_{-\pi}^{\pi}d\xi\,\int_{-\frac{\pi}{2}}^{\frac{\pi}{2}}d\psi\,e^{ 2 \sqrt{\SNR}\repart{\uplinkchannel\hermitian\Y\pilot\conj}}.
	\end{align}
	This detector is computationally very complex since it requires the evaluation of $N$ double integrals which cannot be found in closed form.
	A simpler detector is to find the maximum-likelihood estimate of $\xi$,
	\begin{align}
		\hat{\xi} &= \argmax_{\xi} p(\Y|\pilot,\psi,\xi) \\
		&= -\arg \left[\begin{smallmatrix}
		e^0& e^{-\I\pi\sin\psi}& \dots& e^{-\I\pi(M-1)\sin\psi} \end{smallmatrix}\right]\conj\Y\pilot\conj, 
	\end{align}
	which gives us a maximization of the concentrated likelihood function
	\begin{align}
		\argmax_{\mostlikelypilot} p(\Y|\pilot,\hat{\xi}) 		&= \argmax_{\mostlikelypilot} \int_{-\frac{\pi}{2}}^{\frac{\pi}{2}}d\psi\,e^{ 2 \sqrt{\SNR}\abs{\uplinkchannel\hermitian\Y\pilot\conj}}. \label{eq:no-reciprocity-ml-phase}
	\end{align}
	Further, we can concentrate the likelihood function with respect to $\psi$, corresponding to a joint maximization of $\xi$, $\psi$ and the sequence $\pilot$
	\begin{equation}
		\argmax_{\mostlikelypilot} \max_{\psi} \abs{\uplinkchannel\hermitian\Y\pilot\conj}.
	\end{equation}
	\item The uplink channel is line-of-sight but the uplink signal has been phase shifted such that the first antenna element in the base station has a real-valued channel to the terminal, $\uplinkchannel =  \sqrt{\beta}\left[\begin{smallmatrix}
	e^0& e^{-\I\pi\sin\psi}& \dots& e^{-\I\pi(M-1)\sin\psi} \end{smallmatrix}\right]\transpose$, $p(\psi) = \frac{1}{\pi}, \frac{-\pi}{2} < \psi \leq \frac{\pi}{2}$.
	The detector in this case is 
	\begin{align}
	\argmax_{\mostlikelypilot} p(\Y|\pilot) &= \argmax_{\mostlikelypilot} \int_{-\frac{\pi}{2}}^{\frac{\pi}{2}}d\psi\, p(\Y|\pilot,\psi)p(\psi) \\
	&= \argmax_{\mostlikelypilot} \int_{-\frac{\pi}{2}}^{\frac{\pi}{2}}d\psi\,e^{ 2 \sqrt{\SNR}\repart{\uplinkchannel\hermitian\Y\pilot\conj}}. \label{eq:no-reciprocity-los-without-phase-shift}
	\end{align}
\end{enumerate}
The detector can be extended to encompass different uplink channel models than the ones listed here.

\subsection{Reciprocity Does not Hold, but the Array is Calibrated}\label{sec:calibrated-array}
In the final scenario, the terminal knows the downlink channel. Reciprocity does not hold so $\uplinkchannel$ is independent of $\channel$. However, the array is calibrated so that the uplink channel is also part of a channel set, $\uplinkchannel \in \uplinkchannelset$, i.e., we know the structure of the uplink channel as well.
The terminal transmits an uplink sequence based on the downlink channel and the base station detects the sequence with largest likelihood value
\begin{align}
	\argmax_{\mostlikelypilot} p(\Y|\pilot) &= \argmax_{\mostlikelypilot} \frac{1}{\abs{\uplinkchannelset}}\sum_{\uplinkchannel\in\uplinkchannelset} p(\Y|\pilot,\uplinkchannel) \\
	&= \argmax_{\mostlikelypilot} \sum_{\uplinkchannel\in\uplinkchannelset}\frac{1}{\pi^{M\tau}}e^{-\norm{\Y-\sqrt{\SNR}\uplinkchannel\pilot\transpose}^2} \\
	&= \argmax_{\mostlikelypilot} \sum_{\uplinkchannel\in\uplinkchannelset}e^{2\sqrt{\SNR}\repart{\pilot\transpose\Y\hermitian\uplinkchannel}}.
\end{align}
Additionally, the uplink channel could come from a channel set but also be phase shifted, $\uplinkchannel = e^{\I \nu}\uplinkchannel^{\text{beam}}, \uplinkchannel^\text{beam}\in \uplinkchannelset, -\pi\leq \nu\leq\pi$. Note that this is not the same as $\theta$ being known or unknown.
In this case, the sequence with the largest likelihood value is obtained from
\begin{align}
	\argmax_{\mostlikelypilot} p(\Y|\pilot) &=  \argmax_{\mostlikelypilot} \frac{1}{\abs{\uplinkchannelset}}\sum_{\uplinkchannel^\text{beam}\in\uplinkchannelset}\int_{-\pi}^{\pi}d\nu\frac{1}{2\pi}\,p(\Y|\pilot,\uplinkchannel^\text{beam},\nu)\\
	&=\argmax_{\mostlikelypilot} \sum_{\uplinkchannel^\text{beam}\in\uplinkchannelset} \int_{-\pi}^{\pi} d\nu \, e^{-\norm{\Y-\sqrt{\SNR}e^{\I\nu}\uplinkchannel^\text{beam}\pilot\transpose}^2} \\
	&=\argmax_{\mostlikelypilot} \sum_{\uplinkchannel^\text{beam}\in\uplinkchannelset} \int_{-\pi}^{\pi} d\nu\,e^{2\sqrt{\SNR}\abs{\pilot\transpose\Y\hermitian\uplinkchannel^\text{beam}}\cos\nu} \\
	&=\argmax_{\mostlikelypilot}\sum_{\uplinkchannel^\text{beam}\in\uplinkchannelset} I_{0}\left(2\sqrt{\SNR}\abs{\pilot\transpose\Y\hermitian\uplinkchannel^\text{beam}}\right),
\end{align}
where 
\begin{equation}
I_0(z) = \frac{1}{\pi}\int_{0}^{\pi}e^{z\cos\theta}d\theta
\end{equation}
is the modified Bessel function of the first kind.

\section{Design of Sequence Mappings}\label{sec:reciprocity-channel-sequence-algorithm}
When the terminal has knowledge about the beam that it is in, the sequence mapping that maps the channel beam to the uplink reference signal is crucial for the detection performance.
In Section~\ref{sec:dl-channel-detection}, we have seen the minimization of three different metrics appear in order to minimize the pairwise error probability:
\begin{enumerate}
	\item $	\displaystyle\min_{s:\pilot=s(\channel)} \max_{\substack{n,n' \\ n \neq n'}}\abs{\truepilot\hermitian\falsepilot\truechannel\hermitian\falsechannel} = \min_{s:\pilot=s(\channel)}  \metric{\channelmatrix,\pilotmatrix}{U}$, when reciprocity holds and $\theta$ is unknown,
	\item $\displaystyle	\min_{s:\pilot=s(\channel)} \max_{\substack{n,n' \\ n \neq n'}}\repart{\truepilot\hermitian\falsepilot\truechannel\hermitian\falsechannel} = \min_{s:\pilot=s(\channel)}  \metric{\channelmatrix,\pilotmatrix}{K}$, when reciprocity holds and $\theta$ is known,
	\item $\displaystyle	\min_{s:\pilot=s(\channel)} \max_{\substack{n,n' \\ n \neq n'}}\abs{\truepilot\hermitian\falsepilot} = \min_{s:\pilot=s(\channel)} \metric{\pilotmatrix}{NR}$, when reciprocity does not hold and the uplink channel is Rayleigh fading. This metric minimization problem is the Grassmannian line packing problem \cite{Love08}.
\end{enumerate}
In the cases where reciprocity holds, the structure of $\channelset$ will affect which sequences that are optimal. If all channel pairs have the same inner product, i.e., $\truechannel\hermitian\falsechannel$ is the same for all $n$ and $n'$, the optimal sequence pairs will also have the same inner product. On the other hand, if there is some structure in $\channelset$ such that, for instance, a pair is orthogonal, $\truechannel\hermitian\falsechannel=0$, we can ``get away'' with using the same pilot for the pair, i.e., $\truepilot=\falsepilot$, if we look at this pair in isolation. Structures with more diverse inner products between the channel pairs will be beneficial for our methods. 

In practice, a small pairwise error probability translates into a small mean-square error. For the contrary to be true, certain restrictions on the beams in $\channelset$ should be made. Our analysis holds for any $\channelset$ where $\norm{\channel}^2 = M\beta$, but additional constraints on $\channelset$ can be made, e.g., $\norm{\channel_n - \channel_{n'}}^2 > \varepsilon$, $n\neq n'$, i.e., that the norm-square between two channels in the set should be greater than a threshold $\varepsilon$.

Finding the optimal sequence mapping is an NP-hard problem and the sequences are highly coupled---changing one sequence affects the performance of every pair of sequences which includes the changed sequence. These two properties make it very difficult to find the optimal way of designing sequences with a small metric. 
We propose an algorithm which randomly draws sequences with a certain distribution, that can depend on the set of beams, multiple times and keeps the sequence mapping that so far had the smallest metric. The algorithm is presented in Algorithm~\ref{alg:generate-sequences} and we refer to this algorithm as the Sequence Generation Algorithm (\SequenceGeneratingAlgorithm). 

\begin{algorithm}[tbp]
	\caption{Sequence Generation Algorithm (\SequenceGeneratingAlgorithm), takes the distribution from which to draw sequences and the metric to be minimized as input, and gives the best sequence found as output }
	\label{alg:generate-sequences}	
	\begin{algorithmic}[1]
		\REQUIRE $p(\cdot)$, $\metric{\cdot,\channelmatrix}{}$
		\ENSURE $\pilotmatrix$
		\FOR{ a predefined number of iterations}
		\STATE Create $\newpilotmatrix \sim p(\cdot)$. 
		\STATE Normalize the columns of $\newpilotmatrix$ such that each column $\bm{v}$ has $\norm{\bm{v}} = 1$.
		\IF {$\metric{\newpilotmatrix,\channelmatrix}{} < \metric{\pilotmatrix,\channelmatrix}{}$ }
		\STATE $\pilotmatrix = \newpilotmatrix$
		\ENDIF
		\ENDFOR
	\end{algorithmic}
\end{algorithm}

\begin{algorithm}[tbp]
	\caption{Unknown Phase-Shift Correlation Algorithm (\UnknownAlgorithm), finds a correlation matrix for minimizing the metric $ \metric{\cdot,\cdot}{U}$ }
	\label{alg:correlation-unknown}	
	\begin{algorithmic}[1]
		\REQUIRE $\channelmatrix$
		\ENSURE $\bm{R}$
		\STATE Calculate $\bm{R}^G = \channelmatrix\hermitian\channelmatrix/M$ 
		\STATE $m = \displaystyle\min_{i,j, \abs{\bm{R}^G_{i,j}}\neq 0}  \abs{\bm{R}^G_{i,j}}$
		\STATE Create $\bm{R}^P_H = 1 \oslash (\abs{\bm{R}^G}m)$. Division by zero is treated as $1$. 
		\STATE Let $\bm{R}$ be the positive semidefinite matrix which is closest to $\bm{R}^P_H$. \[ \bm{R} = \argmin_{\bm{R}\succcurlyeq\bm{0}} \norm{ \bm{R}^P_H - \bm{R} } \]
	\end{algorithmic}
\end{algorithm}

\begin{algorithm}[tbp]
	\caption{Known Phase-Shift Correlation Algorithm (\KnownAlgorithm), finds a correlation matrix for minimizing the metric $ \metric{\cdot,\cdot}{K}$  }
	\label{alg:correlation-known}	
	\begin{algorithmic}[1]
		\REQUIRE $\channelmatrix$
		\ENSURE $\bm{R}$
		\STATE Calculate $\bm{R}^G = \channelmatrix\hermitian\channelmatrix/M$ 
		\STATE $m_R = \displaystyle\min_{i,j, \repart{\bm{R}^G_{i,j}}\neq 0} | \repart{\bm{R}^G_{i,j}}|$
		\STATE $m_I = \displaystyle\min_{i,j, \impart{\bm{R}^G_{i,j}}\neq 0} | \impart{\bm{R}^G_{i,j}} |$
		\STATE Create $\bm{R}^P_H = 1 \oslash (\repart{\bm{R}^G}m_R) +  \I  \oslash (\impart{{\bm{R}^G}}m_I)$. Division by zero is treated as $1$ in the real part and $0$ in the imaginary part. This way $\bm{R}^P$ will be Hermitian.
		\STATE Let $\bm{R}$ be the positive semidefinite matrix which is closest to $\bm{R}^P_H$. \[ \bm{R} = \argmin_{\bm{R}\succcurlyeq\bm{0}} \norm{ \bm{R}^P_H - \bm{R} } \]
	\end{algorithmic}
\end{algorithm}

The optimal distribution to use in the algorithm is unknown, therefore, as a baseline, we draw sequences from the $\CN{\bm{0}}{\bm{I}}$ distribution.
In an effort of finding better sequences we have experimented with introducing a correlation of the sequences, i.e., drawing sequences from the distribution $\CN{\bm{0}}{\bm{R}}$, where $\bm{R}$ is the correlation matrix. Two algorithms for finding such $\bm{R}$:s can be found in Algorithms~\ref{alg:correlation-unknown} and~\ref{alg:correlation-known} for the case where the phase shift is unknown and known, respectively. We refer to these algorithms as the Unknown Phase-Shift Correlation Algorithm (\UnknownAlgorithm) and the Known Phase-Shift Correlation Algorithm (\KnownAlgorithm), respectively. Note, that these algorithms are heuristic and better performance than the baseline $\bm{R}=\bm{I}$ cannot be guaranteed. However, from simulations we have seen that especially in the case where the grid-of-beams is based on a DFT grid-of-beams world, see Section~\ref{sec:results} for simulation setup and results, these algorithms perform well, in the sense that with the same number of iterations the metrics that we found with $\CN{\bm{0}}{\bm{R}}$ are generally lower than the ones with $\CN{\bm{0}}{\bm{I}}$.
While our proposed algorithms perform very well, they are not provably optimal and we have to leave the search for improved schemes for future work.

\section{Numerical Results}\label{sec:results}

In this section we evaluate the performance of the detectors and the sequence mappings. We evaluate the performance based on the probability of incorrect detection, $p_\text{E}$, and mean-square error (MSE) between the true channel, and detected beam. Throughout all simulations, we assume $M=10$ antennas at the base station, $N=70$ vectors in $\channelset$, and $\beta=1$ such that $\norm{\channel}^2 = M\beta = 10$ for $\channel\in\channelset$. Whenever we generate sequences with the \SequenceGeneratingAlgorithm, Algorithm~\ref{alg:generate-sequences}, we run for $10^6$ iterations and save the sequences with the lowest metric value. Running the algorithm for more iterations will have a higher probability of generating better sequence mappings. However, the time between two improvements will increase.

First, we introduce the discrete Fourier transform (DFT) grid-of-beams world \cite{Shen16}. This corresponds to the terminal being in line-of-sight, but with a finite number of angles, to a base station with a uniform linear array. The beams in $\channelset$ are 
\begin{equation}
\channel_n = \sqrt{\beta}\begin{bmatrix}
0, & \dots, & e^{2 \pi \I (m-1) \frac{n-1}{N} }, & \dots, & e^{2 \pi \I (M-1) \frac{n-1}{N} }
\end{bmatrix}\transpose, \label{eq:dft-beam}
\end{equation}
for $n = 1,\dots,N$.

In the DFT beam case, the correlation between two beams depends on the ``beam-distance'', $n'-n$, and the ratio, $\frac{N}{M}$, 
\begin{align}
\truechannel\hermitian\falsechannel &=  \beta\begin{bmatrix}
0, \dots  e^{2 \pi \I (m-1) \frac{n-1}{N} }, \dots, e^{2 \pi \I (M-1) \frac{n-1}{N} }
\end{bmatrix}\hermitian\nonumber\\ & \phantom{=}\times\begin{bmatrix}
0, \dots, e^{2 \pi \I (m-1) \frac{n'-1}{N} }, \dots, e^{2 \pi \I (M-1) \frac{n'-1}{N} }
\end{bmatrix} \\
&=\beta\sum_{m=1}^{M}e^{2\pi \I (m-1)\frac{n'-n}{N}}.
\end{align}
Clearly, if $\frac{N}{M}$ is an integer and $n'-n \bmod \frac{N}{M} = 0$, then, for $n'\neq n$,
\begin{equation}
\truechannel\hermitian\falsechannel = \beta\sum_{m=1}^{M}e^{2\pi \I \frac{m-1}{M}} = 0.
\end{equation}
We know that $\metric{\pilotmatrix,\channelmatrix}{U} \geq 0$.
This can be used to determine the sequences to be used for CSI feedback. Hence, if $N > M$, $N/M$ is an integer, and $N > \tau$, we only need $\tau = N/M$ to make 
\begin{equation}
\metric{\pilotmatrix,\channelmatrix}{U}  = 0,
\end{equation} which is optimal in this case, by simply having $N/M$ orthogonal sequences and assigning sequence $\pilot_{(n\bmod\frac{N}{M})+1}$ to $\channel_n$.

We start by studying the cases where reciprocity holds.
In Fig.~\ref{fig:op-on-dft-nps} we show the result of a simulation where the terminal lies in one of the $N$ beams of the DFT grid. The phase shift is known by the terminal and compensated for. The terminal detects the beam and transmits one of $\tau$ orthogonal sequences based on the mapping
\begin{equation}
\pilot_{(n\bmod \tau) +1} = s(\channel_n). \label{eq:orthogonal-pilot-mapping}
\end{equation}
``No CSI'' corresponds to the case where the terminal has no CSI at all and therefore sends the same sequence all the time, i.e., the benchmark scenario considered in Section~\ref{sec:reciprocity-no-csi}. The specifics of the transmitted sequence does not influence performance. It can be any (one) sequence.
We can see that $\tau~=~2$ gives a better performance than $\tau~=~3$ at high SNRs. This comes from the fact that beam 1 and beam 70 will have the same pilot in the case with three pilots and since the two beams are ``close'', in the sense that they have a large correlation, they will often be mistaken for each other. 

In Fig.~\ref{fig:op-on-dft-ps} we show results from the same type of simulation, but in this case the phase shift $\theta$ is unknown by the terminal. Now, the simulations correspond to the case in Section~\ref{sec:reciprocity-only-g}, except for ``No CSI'' which still corresponds to the scenario without CSI in Section~\ref{sec:reciprocity-no-csi}. By comparing the results in Figs.~\ref{fig:op-on-dft-nps} and \ref{fig:op-on-dft-ps} we see that there is a significant degradation in performance of not knowing the phase, especially in the cases where the sequence-channel pairs are bad (in the sense that the corresponding metric is large).

\begin{figure*}
	\centering
	\includegraphics{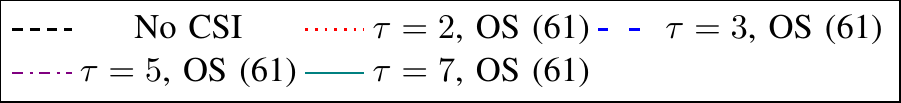} \vspace{1mm}\\
	\subfloat[phase shift $\theta$ known]{
	\centering
	\includegraphics{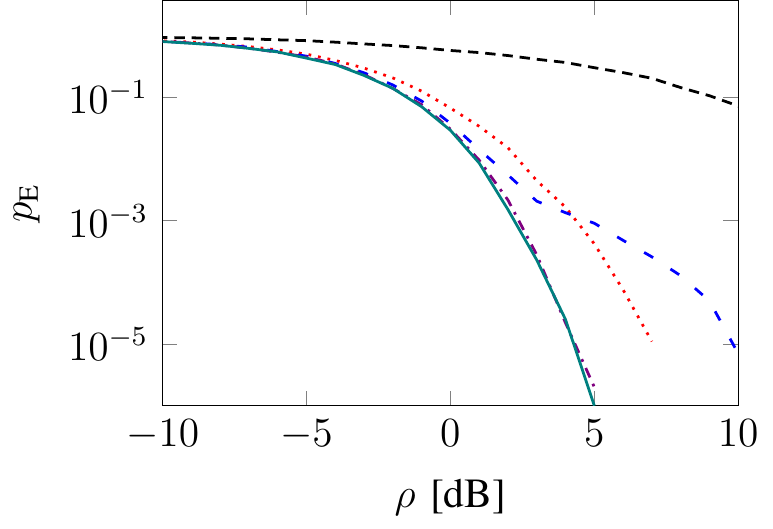}
	\label{fig:op-on-dft-nps}
	}
	\subfloat[phase shift $\theta$ unknown]{
	\centering
	\includegraphics{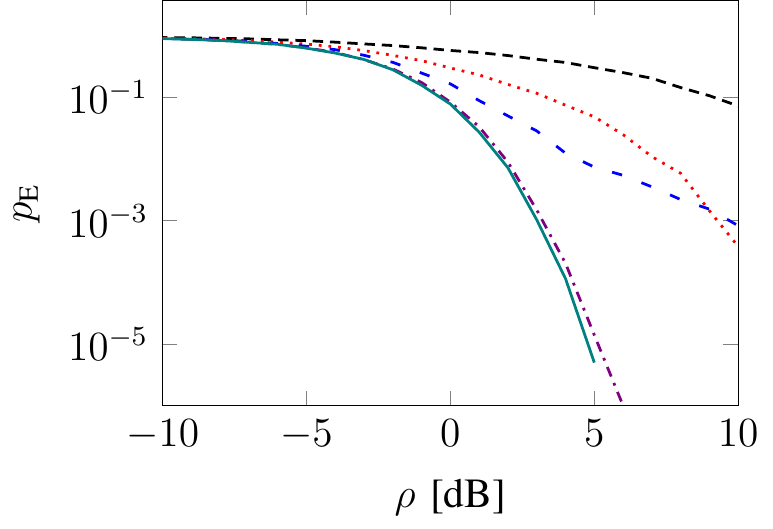}
	\label{fig:op-on-dft-ps}
	}
	\caption{Probability of incorrect detection in a DFT grid-of-beams world, \eqref{eq:dft-beam}, where orthogonal sequences (OS) of length $\tau$ are assigned through mapping \eqref{eq:orthogonal-pilot-mapping}. The terminal is located on the grid. In Fig. \ref{fig:op-on-dft-nps} the terminal knows the phase shift while in Fig. \ref{fig:op-on-dft-ps} it does not. However, in both figures the line titled ``No CSI'' corresponds to the case where the terminal has no channel knowledge, both $\channel$ and $\theta$ are unknown. }
	\label{fig:op-on-dft}
\end{figure*}

In Figs.~\ref{fig:a-on-dft-nps} and \ref{fig:a-on-dft-ps}, we have included sequence mappings generated from the \SequenceGeneratingAlgorithm, Algorithm~\ref{alg:generate-sequences}, in the scenarios with known and unknown phase shifts, respectively.
The sequences generated by the algorithm have length $\tau = 3$. In the sequences corresponding to ``White algorithm'', we used $\CN{\bm{0}}{\bm{I}}$ to draw the sequences and in ``Correlated algorithm'' we used  $\CN{\bm{0}}{\bm{R}}$ to draw the sequences, where $\bm{R}$ is chosen as a function of $\channelmatrix$ with the \KnownAlgorithm, Algorithm~\ref{alg:correlation-known}, and the \UnknownAlgorithm, Algorithm~\ref{alg:correlation-unknown}, respectively. The best sequence mapping, with respect to $\metric{\channelmatrix,\pilotmatrix}{K}$ and $\metric{\channelmatrix,\pilotmatrix}{U}$, respectively, was saved and used for the simulations.
We see that carefully selecting the sequence mapping will hugely impact the performance of the detector; the error probability is lower when using the proposed algorithms than when using orthogonal sequences with $\tau=3$. 
The metric values for the simulations can be found in Table~\ref{tab:metric-values} where we indeed can confirm that cases with lower metric values perform better in terms of detection probability. However, the metric will first come in play at high SNRs when the worst pair dominates the performance. Recall that $\SNR$ is the SNR after coherent integration over the sequence, meaning that the sequence length will not affect the performance other than possibly lowering the metric. Note that, even though the sequences are designed to minimize the pairwise error probability, the performance translates to a small MSE which is a more relevant metric for performance of practical systems. 

\begin{table}
	\centering
	\caption{Metric values, rounded to two decimal points, corresponding to the simulations in Fig.~\ref{fig:a-on-dft}, where the left column, $\metric{\channelmatrix,\pilotmatrix}{K}$, corresponds to Fig.~\ref{fig:a-on-dft-nps} and the right column, $\metric{\channelmatrix,\pilotmatrix}{U}$, corresponds to Fig.~\ref{fig:a-on-dft-ps}.}
	\label{tab:metric-values}
	\begin{tabular}{|l|c|c|}
		\hline
		& $\metric{\channelmatrix,\pilotmatrix}{K}$ & $\metric{\channelmatrix,\pilotmatrix}{U}$ \\ \hline
		\includegraphics{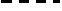} No CSI & - & 9.67 \\ \hline
		\includegraphics{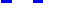} $\tau=3$, OS \eqref{eq:orthogonal-pilot-mapping} & 8.89 & 9.67 \\ \hline
		\includegraphics{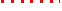} $\tau=3$, White algorithm & 5.41 & 7.73 \\ \hline
		\includegraphics{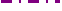} $\tau=3$, Correlated algorithm & 4.99 & 7.11 \\ \hline
		\includegraphics{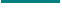} $\tau=7$, OS \eqref{eq:orthogonal-pilot-mapping} & 0 & 0 \\ \hline
	\end{tabular}
\end{table}

\begin{figure*}
	\centering
	\includegraphics{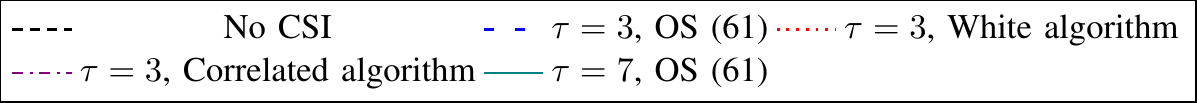} \vspace{1mm}\\
	\subfloat[phase shift $\theta$ known]{
		\centering
		\includegraphics{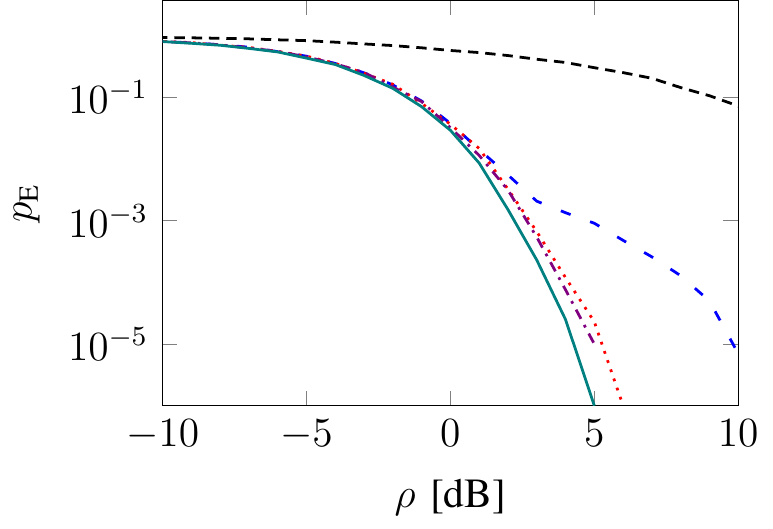}
		\label{fig:a-on-dft-nps}
	}
	\subfloat[phase shift $\theta$ unknown]{
		\centering
		\includegraphics{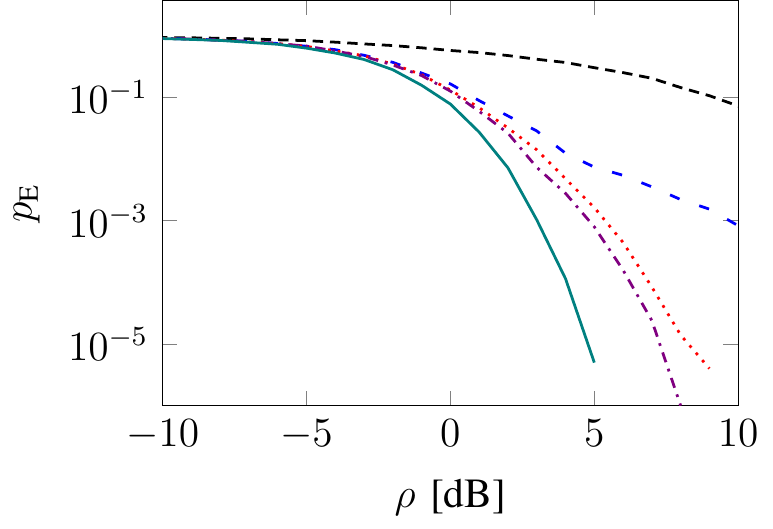}
		\label{fig:a-on-dft-ps}
	}
	\caption{Probability of incorrect detection in a DFT grid-of-beams world, \eqref{eq:dft-beam}, with orthogonal sequences (OS) of length $\tau$ that are assigned through mapping \eqref{eq:orthogonal-pilot-mapping} and sequences generated with the \SequenceGeneratingAlgorithm, Algorithm~\ref{alg:generate-sequences}. with $\CN{\bm{0}}{\bm{I}}$ (White algorithm) and $\CN{\bm{0}}{\bm{R}}$ (Correlated algorithm) where $\bm{R}$ is generated with the \KnownAlgorithm, Algorithm~\ref{alg:correlation-known}, in Fig.~\ref{fig:a-on-dft-nps}, and the \UnknownAlgorithm, Algorithm~\ref{alg:correlation-unknown}, in Fig.~\ref{fig:a-on-dft-ps}. The terminal is located on the grid. In Fig.~\ref{fig:a-on-dft-nps} the terminal knows the phase shift while in Fig.~\ref{fig:a-on-dft-ps} it does not. However, in both figures the line titled ``No CSI'' corresponds to the case where the terminal has no channel knowledge, both $\channel$ and $\theta$ are unknown. }
	\label{fig:a-on-dft}
\end{figure*}

Although we assume that the true downlink channel is from a finite set of beams we can extend the results to a ``continuous beamspace'' by quantizing the true channel to the beam in the set that is closest to the true channel, in the mean-square sense. In Fig.~\ref{fig:a-off-dft} we show the results from a simulation where the true channel is a line-of-sight channel (with an arbitrary angle-of-arrival). In the cases where the terminal knows the channel, it finds the beam in the set that is closest to the true channel and transmits the sequence corresponding to that beam. Other than that, the simulation setup is the same as in Fig.~\ref{fig:a-on-dft}, and the same sequences are used for the beams.
We can again see that carefully selecting the sequence mapping improves the performance. Here, the MSE is lower when using sequences generated from the proposed algorithms than when using orthogonal sequences, while using the same sequence length, $\tau=3$. 
The MSE levels out at high SNRs where the only mismatch that remains is between the true channel and the quantized channel.

\begin{figure*}
	\centering
	\includegraphics{leg-a-off-dft} \vspace{1mm}\\
	\subfloat[phase shift $\theta$ known]{
		\centering
		\includegraphics{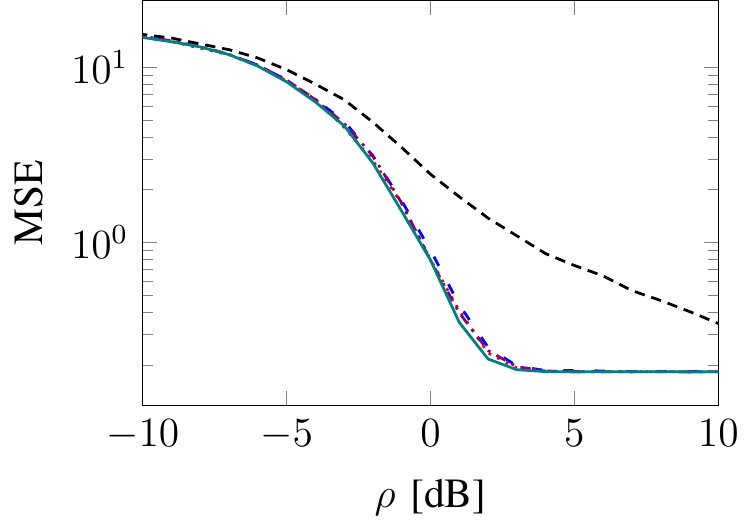}
		\label{fig:a-off-dft-nps}
	}
	\subfloat[phase shift $\theta$ unknown]{
		\centering
		\includegraphics{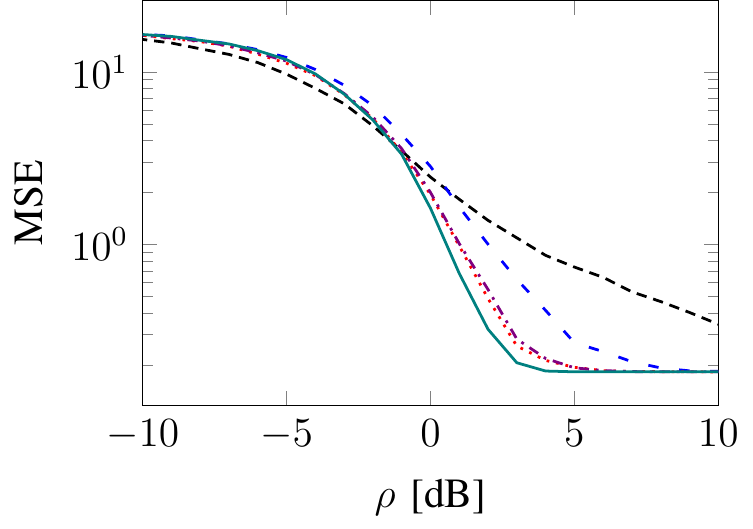}
		\label{fig:a-off-dft-ps}
	}
	\caption{MSE between the true channel and the detected channel. The true channel is a line-of-sight channel (with an arbitrary angle-of-arrival) to the base station. The terminal selects the channel from a DFT grid-of-beams world, \eqref{eq:dft-beam}, which is closest to the true channel, in the mean-square sense.
	Sequences are assigned as follows: Orthogonal sequences (OS) of length $\tau$ are assigned through mapping \eqref{eq:orthogonal-pilot-mapping}; through \SequenceGeneratingAlgorithm, Algorithm~\ref{alg:generate-sequences}, with $\CN{\bm{0}}{\bm{I}}$ (White algorithm); and through \SequenceGeneratingAlgorithm, Algorithm~\ref{alg:generate-sequences}, with $\CN{\bm{0}}{\bm{R}}$ (Correlated algorithm) where $\bm{R}$ is generated with the \KnownAlgorithm, Algorithm~\ref{alg:correlation-known}, in Fig.~\ref{fig:a-off-dft-nps}, and the \UnknownAlgorithm, Algorithm~\ref{alg:correlation-unknown}, in Fig.~\ref{fig:a-off-dft-ps}.
	In Fig.~\ref{fig:a-off-dft-nps} the terminal knows the phase shift while in Fig.~\ref{fig:a-off-dft-ps} it does not. However, in both figures the line titled ``No CSI'' corresponds to the case where the terminal has no channel knowledge, both $\channel$ and $\theta$ are unknown. }
	\label{fig:a-off-dft}
\end{figure*}

Further, we also simulate the more practical scenario when the channel estimate at the terminal is imperfect. The simulation scenario is the same as in Fig.~\ref{fig:a-off-dft}. However, the terminal has estimated the angle-of-arrival with an error of $\mathcal{N}(0, 0.1/(\SNR\cos^2\phi))$, where $\phi$ is the true angle of arrival, i.e., it scales with $\SNR$. This choice of $0.1$ corresponds to the Cramér-Rao lower bound when the downlink SNR is approximately 22 dB \emph{lower} than $\SNR$ \cite{Stoica89}, which means that in practice we will have much better angle estimates. The result is presented in Fig.~\ref{fig:dl-error-crlb-angle-0.1}. We can see that even though the channel estimates are imperfect, the proposed methods have a gain over the ``No CSI'' case. However, if the angle-of-arrival estimation error at the terminal is larger, this gain vanishes. Hence, our proposed methods should be used when the downlink channel estimates are accurate.
	
\begin{figure*}
	\centering
	\includegraphics{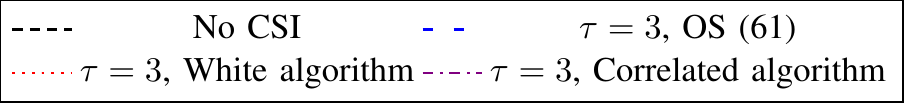} \vspace{1mm}\\
	\subfloat[phase shift $\theta$ known]{
		\centering
		\includegraphics{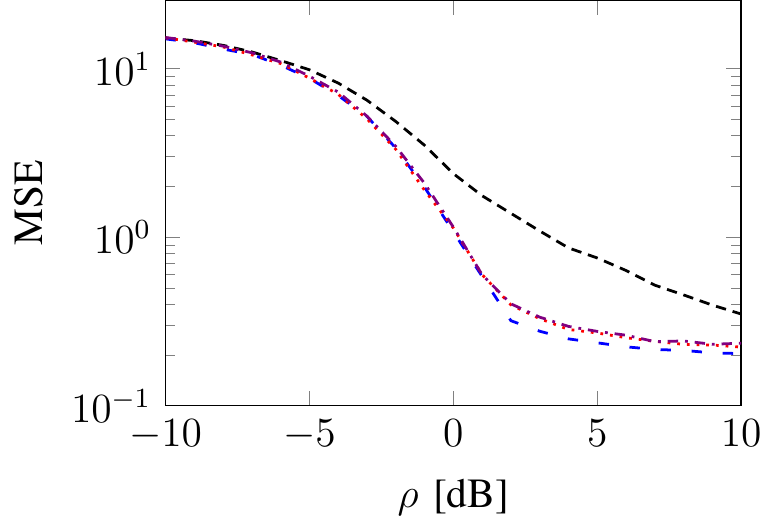}
		\label{}
	}
	\subfloat[phase shift $\theta$ unknown]{
		\centering
		\includegraphics{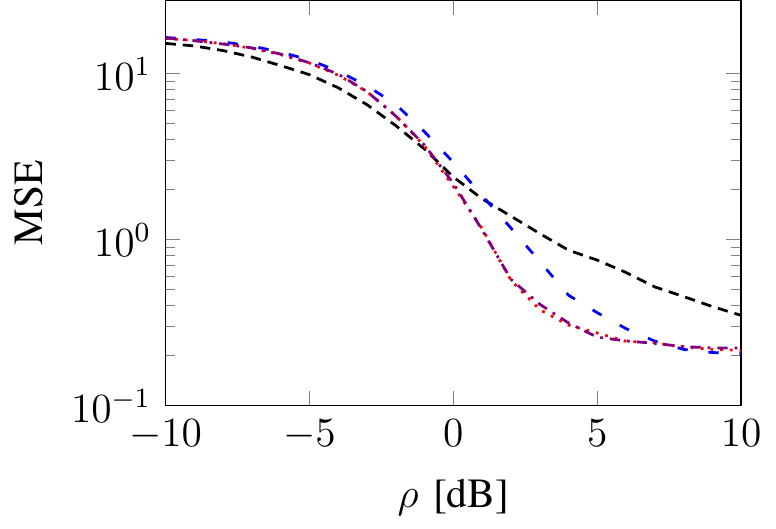}
		\label{}
	}
	\caption{ The same simulation scenario as in Fig. \ref{fig:a-off-dft} but there is an error in the estimated angle-of-arrival at the terminal of $\mathcal{N}(0, 0.1/(\SNR\cos^2\phi))$ degrees. }
	\label{fig:dl-error-crlb-angle-0.1}
\end{figure*}

In the next simulation setup, we study a grid-of-beams world where the beams correspond to a packing in the Grassmannian manifold \cite{Love08}, generated through algorithms in \cite{Dhillon08}. The results can be seen in  Figs.~\ref{fig:a-on-g-nps} and~\ref{fig:a-on-g-ps} with known and unknown phase shifts, respectively. Generally, the metric values are smaller in the Grassmannian grid-of-beams world than in the DFT grid-of-beams worlds since the packing is designed such that there is a small correlation between the beams. Hence, the choice of sequence mapping is less important in this scenario than in the DFT grid-of-beams scenario, at least at the error probabilities studied in the simulations. However, an improvement of more than 1~dB can be seen at error probability $10^{-4}$ in both the cases with and without the knowledge of the phase shift $\theta$ at the terminal.

\begin{figure*}
	\centering
	\includegraphics{leg-a-off-dft} \vspace{1mm}\\
	\subfloat[phase shift $\theta$ known]{
		\centering
		\includegraphics{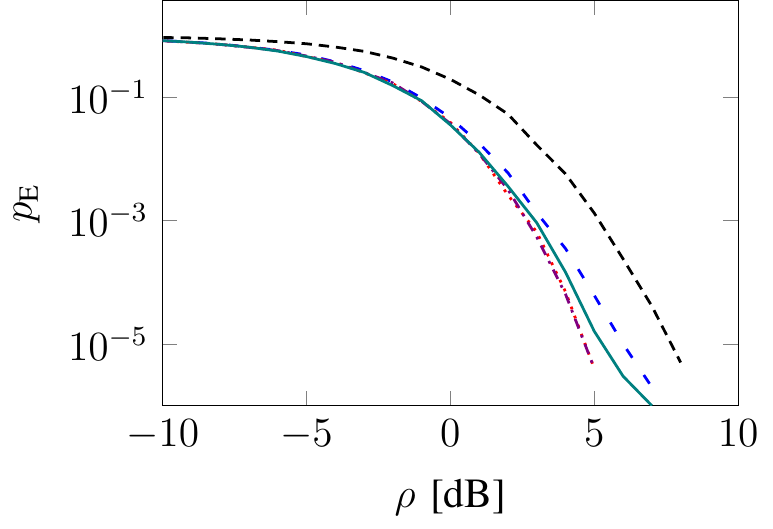}
		\label{fig:a-on-g-nps}
	}
	\subfloat[phase shift $\theta$ unknown]{
		\centering
		\includegraphics{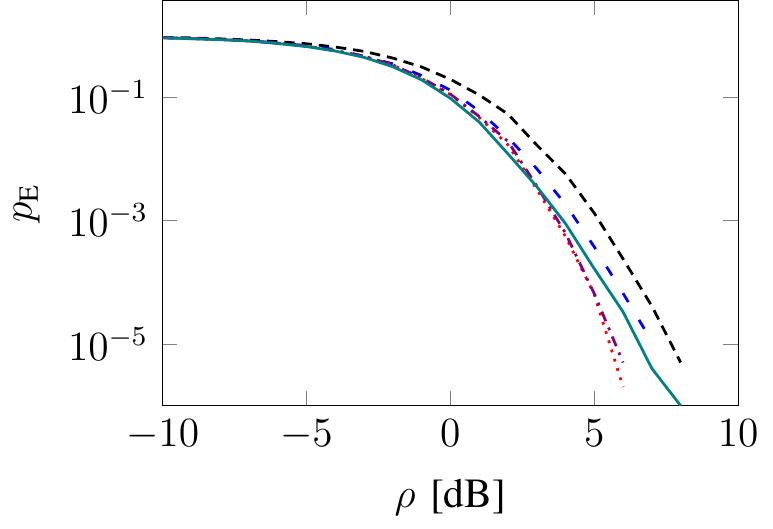}
		\label{fig:a-on-g-ps}
	}
\caption{Probability of incorrect detection in a Grassmannian grid-of-beams world, with orthogonal sequences (OS) of length $\tau$ that are assigned through mapping \eqref{eq:orthogonal-pilot-mapping} and sequences generated with the \SequenceGeneratingAlgorithm, Algorithm~\ref{alg:generate-sequences}, with $\CN{\bm{0}}{\bm{I}}$ (White algorithm) and $\CN{\bm{0}}{\bm{R}}$ (Correlated algorithm) where $\bm{R}$ is generated with the \KnownAlgorithm, Algorithm~\ref{alg:correlation-known}, in Fig.~\ref{fig:a-on-g-nps}, and the \UnknownAlgorithm, Algorithm~\ref{alg:correlation-unknown}, in Fig.~\ref{fig:a-on-g-ps}. The terminal is located on the grid. In Fig.~\ref{fig:a-on-g-nps} the terminal knows the phase shift while in Fig.~\ref{fig:a-on-g-ps} it does not. However, in both figures the line titled ``No CSI'' corresponds to the case where the terminal has no channel knowledge, both $\channel$ and $\theta$ are unknown. }
\label{fig:a-on-g}
\end{figure*}

We also study scenarios where the channel is Rician fading, $\channel=\channel_\text{LoS}+\channel_\text{NLoS}$, i.e., consisting of a free-space line-of-sight component, $\channel_\text{LoS}$, which is known at the terminal, and a non-line-of-sight component which is i.i.d. Rayleigh fading, $\channel_\text{NLoS} \sim\CN{\bm{0}}{\sigma^2\bm{I}}$, and unknown at the terminal. 
Specifically, Fig. \ref{fig:rice-off} shows results for the case where  $\sigma^2=0.1$. Here, the line-of-sight component is 10 times stronger than the non-line-of-sight component. The terminal quantizes $\channel_\text{LoS}$ to the closest beam in the DFT grid-of-beams and transmits the corresponding sequence (same as in Fig.~\ref{fig:a-off-dft}) in the uplink. Note that, in this case, in addition to not having $\channel\in\channelset$, we do not conform completely to our system model in Section~\ref{sec:system-model} where we assumed that $\norm{\channel}^2 = M\beta$, which is not the case here, instead we have $\EX\left\{\norm{\channel}^2\right\} = M(\beta+\sigma^2)$, but Fig.~\ref{fig:rice-off} shows that the developed algorithms can nevertheless be applied. Here, we have two sources of error, the quantization error and the non-line-of-sight component. We can conclude that this channel model and beam set contain enough structure to improve the channel estimates in the uplink.

\begin{figure}
	\centering
	\includegraphics{leg-a-off-dft} \vspace{1mm}\\
	\subfloat[phase shift $\theta$ known]{
		\centering
		\includegraphics{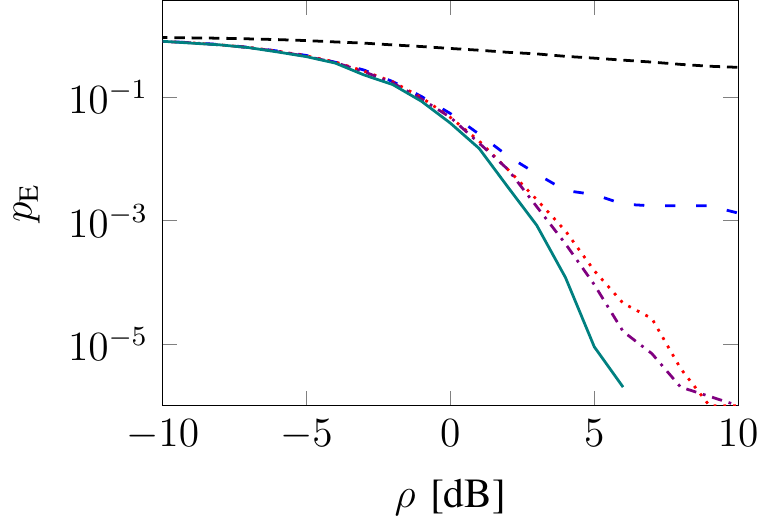}
		\label{}
	}
	\subfloat[phase shift $\theta$ unknown]{
		\centering
		\includegraphics{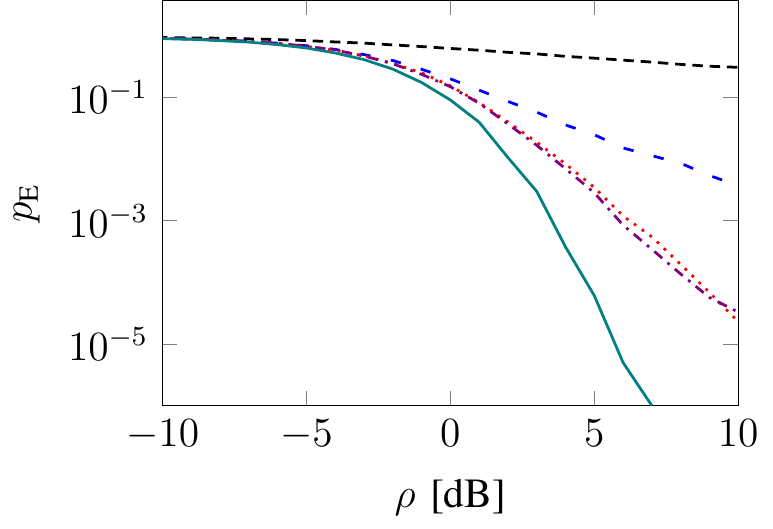}
		\label{}
	}\\
	\vspace*{10pt}
	\hspace*{1.2pt}
	\subfloat[phase shift $\theta$ known]{
		\centering
		\includegraphics{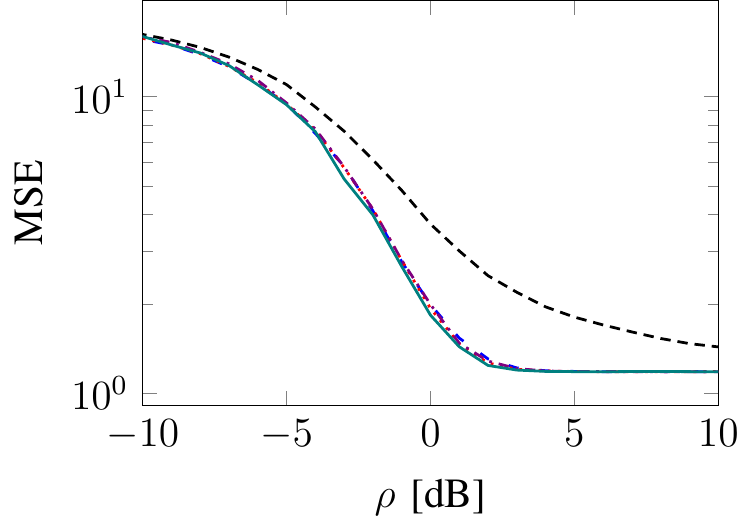}
		\label{}
	}
	\hspace*{1pt}
	\subfloat[phase shift $\theta$ unknown]{
		\centering
		\includegraphics{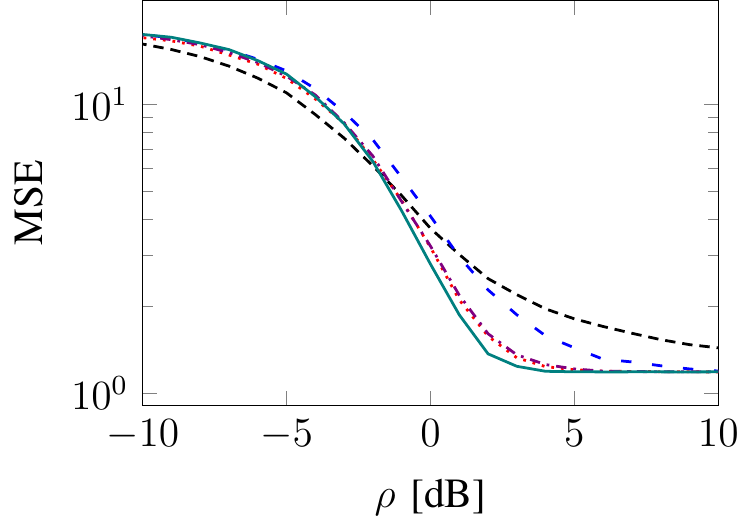}
		\label{}
	}
	\caption{ Simulated Rician fading channel where the line-of-sight component is a line-of-sight channel with arbitrary angle-of-arrival and known perfectly by the terminal. The terminal acts as in Fig.~\ref{fig:a-off-dft}. }
	\label{fig:rice-off}
\end{figure}

Finally, we study the cases where reciprocity does not hold. We evaluate the performance in two different scenarios: the uplink channel is line-of-sight with a phase shift and the uplink channel is line-of-sight without a phase shift. In both cases, we evaluate the three estimators, \eqref{eq:no-reciprocity-rayleigh}, \eqref{eq:no-reciprocity-ml-phase}, and \eqref{eq:no-reciprocity-los-without-phase-shift}, introduced in Section~\ref{sec:no-reciprocity}. We consider the cases where the sequences are orthogonal, requiring $\tau=70$, and where we have generated sequences of length $\tau=3$ through the \SequenceGeneratingAlgorithm, Algorithm~\ref{alg:generate-sequences}, with $\CN{\bm{0}}{\bm{I}}$. The results of the simulations can be seen in Figs.~\ref{fig:no-rec-los-nps} and~\ref{fig:no-rec-los-ps} for the cases where the uplink channel is line-of-sight without and with a phase shift, respectively. In Fig.~\ref{fig:no-rec-los-nps} and Fig.~\ref{fig:no-rec-los-ps} the detectors \eqref{eq:no-reciprocity-los-without-phase-shift} and \eqref{eq:no-reciprocity-ml-phase} perform best, respectively. However, \eqref{eq:no-reciprocity-rayleigh} also performs relatively well, and is computationally much simpler that the others. Note that, $\tau=70$ is analogous to feeding back the precoder codebook index which best fits the true channel with orthogonal resources, i.e., the benchmark in the case without reciprocity.

We can determine the gain of using reciprocity in the CSI feedback by comparing Fig.~\ref{fig:a-on-dft} and Fig.~\ref{fig:no-rec-los}. From these results we can conclude that when reciprocity is not used, much longer sequences are needed to achieve a comparable performance to the scenarios where reciprocity holds and is taken advantage of. Hence, the conclusion can be drawn, that using reciprocity in the CSI feedback beats conventional CSI feedback strategies.

\begin{figure*}
	\centering
	\includegraphics{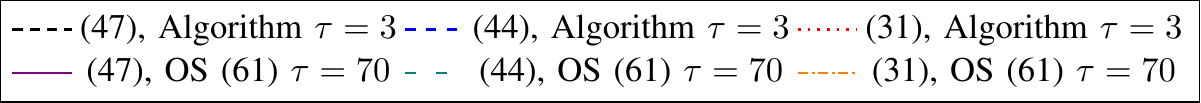} \vspace{1mm} \\
	\subfloat[line-of-sight without phase shift
	]{
		\centering
		\includegraphics{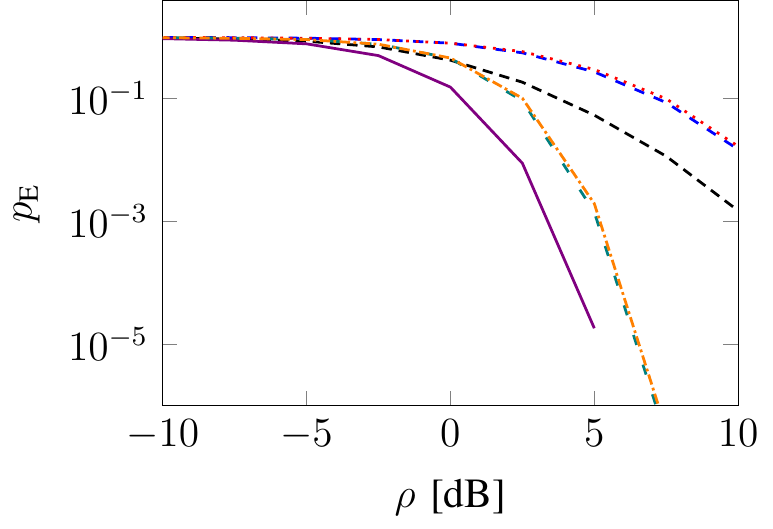}
		\label{fig:no-rec-los-nps}
	}
	\subfloat[line-of-sight with phase shift
	]{
		\centering
		\includegraphics{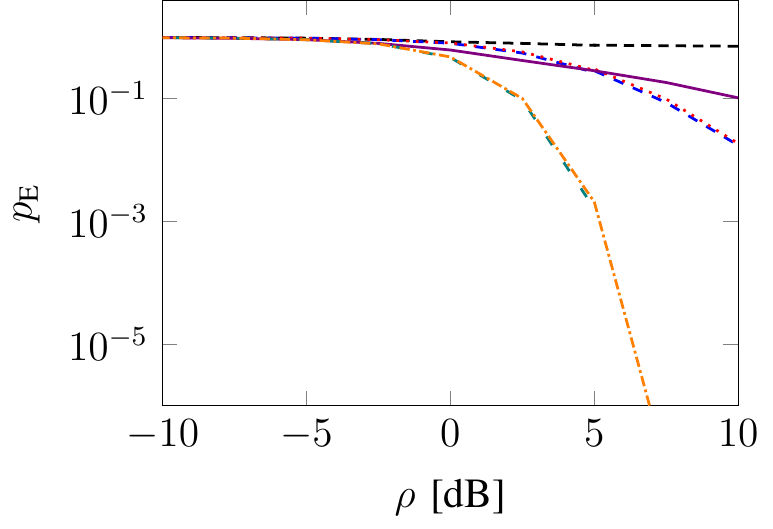}
		\label{fig:no-rec-los-ps}
	}
	\caption{Probability of incorrect detection when reciprocity does not hold and the uplink channel is line-of-sight. Orthogonal sequences (OS) of length $\tau=70$ that are assigned through mapping \eqref{eq:orthogonal-pilot-mapping} and sequences generated with the \SequenceGeneratingAlgorithm, Algorithm~\ref{alg:generate-sequences}, with $\CN{\bm{0}}{\bm{I}}$ are compared.  In Fig.~\ref{fig:no-rec-los-nps} the uplink channel does not contain the phase shift from the unknown distance to the base station while in Fig.~\ref{fig:no-rec-los-ps} it does. }
	\label{fig:no-rec-los}
\end{figure*}

We do not show simulations with a calibrated array, i.e., when the uplink channel is from a set of finite beams, $\uplinkchannel\in\uplinkchannelset$, from Section~\ref{sec:calibrated-array}, since we can draw similar conclusions as in the case in Section~\ref{sec:no-reciprocity}. Additionally, the gain of having a calibrated array is small.

\section{Conclusions}
In this paper, we studied the task of downlink channel detection in a world with a finite number of channels contained in the channel set $\channelset$. 
Such a world can be useful to study, e.g.,  when grid-of-beams-based algorithms are used for channel estimation, i.e., beam sweeping type algorithms. It can also be viewed as a special case of the framework in \cite{Hoon12,Adhikary13}. 
The paper is motivated by claims that reciprocity is not sufficient to obtain accurate channel estimates at the base station when the uplink SNR is low. We address this issue by developing a scheme that combines reciprocity and CSI feedback. The sequence transmitted in the uplink is a function of which of the beams in the channel set that corresponds to the true channel. We have found elegant metrics that should be minimized when designing the sequence mapping in order to minimize the pairwise error probability. 
We also proposed a simple algorithm that can find good sequence mappings.
We have shown that:
\begin{itemize}
	\item In the case where reciprocity holds and the terminal has access to accurate downlink channel estimates, one should use our novel scheme to greatly improve the detection accuracy. Hence, our proposed scheme performs better than the conventional pilot based channel estimation.
	\item When reciprocity holds, but the uplink channel is of poor quality, one should use our scheme and utilize the reciprocity. Our proposed scheme performs better than the conventional way of feeding back the quantized channel   estimates since we can use much shorter uplink sequences to achieve the same performance.
\end{itemize}
In conclusion, whenever uplink-downlink channel reciprocity holds, one should always exploit~it.


\FloatBarrier
\appendices
\section{Proof of Proposition~\ref{prop:pep-theta-unknown} and Proposition~\ref{prop:pep-partial-csi}}\label{app:proof-pep-theta-unknown}
We will prove that the pairwise error probability in the case when reciprocity holds and the terminal knows the downlink channel $\channel$ but $\theta$ is unknown decreases with decreasing $\abs{\truepilot\hermitian\falsepilot\truechannel\hermitian\falsechannel}$ and increasing $\sqrt{\SNR M \beta}$. If the true channel-sequence pair is $(\truechannel,\truepilot)$, such that $\Y=\sqrt{\SNR}e^{\I\theta}\truechannel\truepilot\transpose+\noise$, we make an error in favor of $(\falsechannel,\falsepilot)$ with probability
\begin{align}
	&\Pr\left(\abs{\falsepilot\transpose\Y\hermitian\falsechannel} > \abs{\truepilot\transpose\Y\hermitian\truechannel}\right)\\
	&=\Pr\left(\abs{\sqrt{\SNR}e^{-\I\theta}\falsepilot\transpose\truepilot\conj\truechannel\hermitian\falsechannel+\falsepilot\transpose\noise\hermitian\falsechannel} > \abs{\sqrt{\SNR}e^{-\I\theta}M\beta+\truepilot\transpose\noise\hermitian\truechannel}\right)\\
	&=\Pr\left(\abs{ \sqrt{\SNR M\beta}\truepilot\hermitian\falsepilot\ntruechannel\hermitian\nfalsechannel+e^{\I\theta}\falsepilot\transpose\noise\hermitian\nfalsechannel} > \abs{ \sqrt{\SNR M\beta}+ e^{\I\theta}\truepilot\transpose\noise\hermitian\ntruechannel} \right),\label{eq:normalizedapproach}
\end{align}
where we let $\ntruechannel = \frac{\truechannel}{\sqrt{M\beta}}$ and $\nfalsechannel = \frac{\falsechannel}{\sqrt{M\beta}}$.

Further, let 
\begin{equation}
	\alpha = \truepilot\hermitian\falsepilot \ntruechannel\hermitian\nfalsechannel, \label{eq:startshortcut}
\end{equation} 
\begin{equation}
x = e^{\I\theta}\truepilot\transpose\noise\hermitian\ntruechannel,
\end{equation}
and
\begin{equation}
w = e^{\I\theta}\falsepilot\transpose\noise\hermitian\nfalsechannel.
\end{equation}
The correlation between $x$ and $w$ is 
\begin{align}
\EX\left\{ xw\conj \right\} &= \EX\left\{ \left(e^{\I\theta}\truepilot\transpose\noise\hermitian\ntruechannel\right) \left( e^{\I\theta}\falsepilot\transpose\noise\hermitian\nfalsechannel \right)\conj \right\} \\
&= \EX\left\{  \sum_{i=1}^{\tau}\sum_{j=1}^{M} \sum_{i'=1}^{\tau}\sum_{j'=1}^{M} (\truepilot)_i (\falsepilot\conj)_{i'}  (\ntruechannel)_{j} (\nfalsechannel\conj)_{j'}  \noise_{j,i}\conj  \noise_{j',i'} \right\} \end{align}
\begin{align} 
&=  \sum_{i=1}^{\tau}\sum_{j=1}^{M} (\truepilot)_i (\falsepilot\conj)_{i}  (\ntruechannel)_{j} (\nfalsechannel\conj)_{j}   \\
&=\falsepilot\hermitian\truepilot  \ntruechannel\transpose\nfalsechannel\conj \\
&= \alpha\conj.
\end{align}
Hence, we have a bivariate circularly symmetric Gaussian random variable
\begin{equation}
\begin{bmatrix} x \\ w\end{bmatrix} \sim \CN{\begin{bmatrix}	0\\0\end{bmatrix}}{ \begin{bmatrix} 1 & \alpha\conj \\ \alpha & 1\end{bmatrix}}.\label{eq:mtxxw}
\end{equation}
Let us decompose $w$ as
\begin{equation}
w = ax + by, \label{eq:splitw}
\end{equation}
where $y \sim \CN{0}{1}$ is independent of $x$. We are looking for $a$ and $b$. First note from \eqref{eq:mtxxw} that $\EX\left\{ w | x \right\}~=~\alpha x$ \cite[Th. 10.2]{kay93estimation}. But from \eqref{eq:splitw}, $\EX\left\{ w | x \right\} = ax$. So, $a = \alpha$. Next, we know that $\Var\left\{ w \right\} = \abs{a}^2+\abs{b}^2 = 1 $. So, $\abs{b} = \sqrt{1-\abs{a}^2} = \sqrt{1-\abs{\alpha}^2}$. Therefore, 
\begin{equation}
w = \alpha x + \sqrt{1-\abs{\alpha}^2} y. \label{eq:endshortcut}
\end{equation}
We use this to rewrite  \eqref{eq:normalizedapproach}
\begin{align}
&\Pr\left( \abs{\sqrt{\SNR M\beta}\truepilot\hermitian\falsepilot\ntruechannel\hermitian\nfalsechannel+e^{\I\theta}\falsepilot\transpose\noise\hermitian\nfalsechannel}  >  \abs{\sqrt{\SNR M\beta} + e^{\I\theta}\truepilot\transpose\noise\hermitian\ntruechannel}  \right) \\
&= \Pr\left( \abs{\sqrt{\SNR M\beta} \alpha + \alpha x + \sqrt{1-|\alpha|^2}y}  > \abs{\sqrt{\SNR M\beta} + x}\right)\\
&\overset{(a)}{=}\Pr\left( \abs{|\alpha|(\sqrt{\SNR M\beta}  +  x) + \sqrt{1-|\alpha|^2}y}^2  > \abs{\sqrt{\SNR M\beta} + x}^2\right) \\
&\overset{(b)}{=} \Pr\left(  (1-|\alpha|^2)|y|^2 + 2|\alpha|\sqrt{1-|\alpha|^2}\repart{X\conj y} > (1-|\alpha|^2)\abs{X}^2 \right) \\
&= \Pr\left(  |y|^2 + 2\frac{|\alpha|}{\sqrt{1-|\alpha|^2}}\repart{X\conj y} > \abs{X}^2 \right) \label{eq:simulated} 
\end{align}
where in (a) we used that $y$ is circularly symmetric and in (b) we introduced the random variable $X \sim \CN{\sqrt{\SNR M\beta}}{1}$.
Let $A = \frac{\abs{\alpha}}{\sqrt{1-\abs{\alpha}^2}}$. We continue with the calculations of the pairwise error probability
\begin{align}
& \Pr\left( \abs{y}^2 + 2A\repart{X\conj y} > \abs{X}^2 \right) \\
&\overset{(a)}{=}\Pr\left( \abs{y}^2 + 2A\abs{X}\repart{y} > \abs{X}^2 \right) \\
&\overset{(b)}{=}\frac{1}{2\pi}\int_{-\pi}^{\pi}d\theta_y\,\Pr\left( R_y^2 + 2g\abs{X}R_y > \abs{X}^2 \right)
\\
&=\frac{1}{2\pi}\int_{-\pi}^{\pi}d\theta_y\,\Pr\left(\left( R_y+\abs{X}\left(g-\sqrt{g^2 +1}\right) \right)\underbrace{\left(R_y+\abs{X}\left(g+\sqrt{g^2 +1}\right)\right)}_{>0}>0\right)\\
&=\frac{1}{2\pi}\int_{-\pi}^{\pi}d\theta_y\,\Pr\left(\left( R_y+\abs{X}\left(g-\sqrt{g^2 +1}\right) \right)>0\right)\\
&\overset{(c)}{=} \frac{1}{2\pi}\int_{-\pi}^{\pi}d\theta_y\, \Pr\left(R_y - B\abs{X}>0\right)d\theta_y 
\\
%
&\overset{(d)}{=} \frac{1}{2\pi}\int_{-\pi}^{\pi}d\theta_y\, \int_\complexs dt\, \Pr\left(R_y^2 > B^2\abs{t}^2\right)p_x(t-\sqrt{\eSNR})  \\
&\overset{(e)}{=} \frac{1}{2\pi}\int_{-\pi}^{\pi}d\theta_y\, \int_\complexs dt\, \exp\left(-B^2\abs{t}^2\right) p_x(t-\sqrt{\eSNR}) \\
%
%
&= \frac{1}{2\pi}\frac{1}{\pi}\int_{-\pi}^{\pi}d\theta_y\, \int_\complexs dt\, \exp\left(-B^2\abs{t}^2-\abs{t-\sqrt{\eSNR}}^2\right)  \\
%
&= \frac{1}{2\pi}\int_{-\pi}^{\pi}d\theta_y\,\frac{\exp\left(\frac{\eSNR}{B^2+1} -\eSNR\right)}{B^2+1} \underbrace{\int_\complexs dt\, \frac{B^2+1}{\pi}\exp\left(-\left(B^2+1\right)\abs{t-\frac{\sqrt{\eSNR}}{B^2+1}}^2\right)}_{=1} 
\end{align}
\begin{align}
&\overset{(f)}{=} \frac{e^{-\eSNR}}{2\pi}\int_{-\frac{\pi}{2}}^{\frac{\pi}{2}}d\theta_y\,\left(\frac{1}{B^2+1}e^{\eSNR\frac{1}{B^2+1}} + \frac{B^2}{B^2+1}e^{\eSNR\frac{B^2}{B^2+1} }\right)  \\
&\overset{(g)}{=} \frac{e^{-\eSNR}}{2\pi}\int_{-\frac{\pi}{2}}^{\frac{\pi}{2}}d\theta_y\,\underbrace{\left(ze^{\eSNR z} + (1-z)e^{\eSNR (1-z) }\right)}_{=f(z)} \label{eq:approach2}
\end{align}
where in $(a)$ we use that $y$ is circularly symmetric, i.e., $e^{\I\psi}y$ has the same distribution as  $y$ for any $\psi$, in $(b)$ we marginalize over the phase of $y$ which is a uniform random variable and introduce the random variable $R_y$, the magnitude of $y$ and $g = A\cos\theta_y$, in $(c)$ we introduce $B = \sqrt{g^2 +1}-g$, in $(d)$ we marginalize over the random variable $X$, $p_x(\cdot)$ is the probability distribution of a standard circularly symmetric Gaussian random variable, in $(e)$ we used the fact that $R_y^2$ is exponentially distributed, in $(f)$ we used that $\frac{1}{B}~=~\sqrt{g^2 + 1}+g$ and in $(g)$ we introduced $z~=~\frac{1}{B^2+1}, B^2~=~\frac{1-z}{z}, \frac{B^2}{B^2+1} = 1-z$.

We study the derivative of $f(z)$, the integrand of \eqref{eq:approach2},
\begin{align}
\frac{df(z)}{dz} &= \eSNR z e^{\eSNR z} + e^{\eSNR z} -\eSNR(1-z)e^{\eSNR(1-z)}-e^{\eSNR(1-z)}\\
&= e^{\eSNR z}-e^{\eSNR(1-z)} + \eSNR z\left( e^{\eSNR z} - \underbrace{\frac{1-z}{z}}_{<1}e^{\eSNR(1-z)}\right) \\
&> \underbrace{\left(e^{\eSNR z}-e^{\eSNR(1-z)}\right)}_{>0}(1+\eSNR)  \overset{(a)}{>}0,
\end{align}
where in (a) we used that 
\begin{align}
e^{\eSNR z}-e^{\eSNR (1-z)} &= e^{\eSNR z}\left(1 - e^{\eSNR(1-2z)}\right)>0,
\end{align}
since $z > \frac{1}{2}$.
The derivative of $f(z)$ is always positive meaning that $f(z)$ will increase when increasing $z$ (or $\abs{\alpha}$) for all values where $\cos\theta>0$ i.e., over the whole integration interval in \eqref{eq:approach2}. Since all the function values in the integral increase when $\abs{\alpha}$ increases the pairwise error probability will also increase when $\abs{\alpha}$ increases.

Note that for the case where the terminal does not know $\channel$ it will send the same sequence in every beam, $\pilot_n = \pilot$ for $n=1,\dots,N$, hence we can use this result for both cases.

Using results in \cite[App.~B]{ProakisDigComm}, the pairwise error probability can also be found in closed form, in terms of the Marcum Q-function and Bessel functions. However, using that closed-form expression does not simplify the proof that the pairwise error probability increases in~$\abs{\alpha}$.

\section{Proof of Proposition~\ref{prop:pep-full-csi}}\label{app:proof-pep-theta-known}
We will prove that the pairwise error probability in the case when reciprocity holds and the terminal knows both the downlink channel $\channel$ and the phase shift $\theta$ decreases with decreasing $\repart{\truepilot\hermitian\falsepilot\truechannel\hermitian\falsechannel}$ and increasing $\sqrt{\SNR M \beta}$. If the true channel-sequence pair is $(\truechannel,\truepilot)$, such that $\Y=\sqrt{\SNR}\truechannel\truepilot\transpose+\noise$, the detector makes an error in favor of $(\falsechannel,\falsepilot)$ with probability
\begin{align}
	&\Pr\left(\repart{\falsepilot\transpose\Y\hermitian\falsechannel} > \repart{\truepilot\transpose\Y\hermitian\truechannel}\right)\\
	&=\Pr\left( \repart{\Tr\left\{ \Y\hermitian\left( \falsechannel\falsepilot\transpose-\truechannel\truepilot\transpose \right) \right\}} > 0 \right) \\
	&=\Pr\left( \sqrt{\SNR}\repart{\underbrace{\Tr\left\{\truepilot\conj\truechannel\hermitian\left( \falsechannel\falsepilot\transpose-\truechannel\truepilot\transpose \right)\right\}}_{=M\beta(\alpha-1)}} + \repart{\underbrace{\Tr\left\{ \noise\hermitian\left(\falsechannel\falsepilot\transpose-\truechannel\truepilot\transpose\right) \right\}}_{\sim \CN{0}{2M\beta(1-\repart{\alpha})}}}>0 \right)\\
	&=\Pr\left( \sqrt{\SNR}M\beta(\repart{\alpha}-1)+\mathcal{N}\left(0,M\beta(1-\repart{\alpha})\right) > 0\right)\\
&=Q\left(\sqrt{\SNR M\beta(1-\repart{\alpha})}\right)\\
&=Q\left(\sqrt{\SNR (M\beta-\repart{\truepilot\hermitian\falsepilot\truechannel\hermitian\falsechannel})}\right),\label{eq:pep-theta-known}
\end{align}
where $\alpha$ is defined as in \eqref{eq:startshortcut} and $Q\left(\cdot\right)$ is the Q-function.
Thus, by studying \eqref{eq:pep-theta-known} we see~that~the pairwise error probability decreases when $\repart{\truepilot\hermitian\falsepilot\truechannel\hermitian\falsechannel}$ decreases or when $\sqrt{\SNR M \beta}$ \mbox{increases}.

Note, that in the case where the terminal does not know the downlink channel $\channel$ is the same as when $\pilot_n = \pilot$ for $n=1,\dots,N$. Hence, the pairwise error probability is 
\begin{equation}
	Q\left(\sqrt{\SNR (M\beta-\repart{\truechannel\hermitian\falsechannel})}\right)
\end{equation}
in that case.

\section{Proof that \eqref{eq:no-rec-pep} increases with $\abs{\alpha}$}\label{app:proof3}
We will prove that the pairwise error probability in the case where reciprocity does not hold and the uplink channel is i.i.d. Rayleigh fading decreases with decreasing $\abs{\alpha} = \abs{\truepilot\hermitian\falsepilot}$.
Starting from \eqref{eq:no-rec-pep}, we have
\begin{align}
&\Pr\left( \norm{\abs{\alpha}\bm{x}+\sqrt{1-\abs{\alpha}^2}\bm{y}}^2 > \norm{\bm{x}}^2  \right) \\
&=\Pr\left( -\norm{\bm{x}}^2 + \norm{\bm{y}}^2 + A\bm{x}\hermitian\bm{y} + A\bm{y}\hermitian\bm{x}  > 0 \right) \\
&=\Pr\left(\sum_{m=1}^M -\abs{x_m}^2 + \abs{y_m}^2 + Ax_m\conj y_m + Ay_m\conj x_m > 0 \right),
\end{align}
where $\bm{x}\sim\CN{\bm{0}}{(\rho\beta+1)\bm{I}}$, $\bm{y}\sim\CN{\bm{0}}{\bm{I}}$ and $A =\frac{\abs{\alpha}}{\sqrt{1-\abs{\alpha}^2}}$. We only consider $\abs{\alpha} = \abs{\truepilot\hermitian\falsepilot}< 1$, i.e., $\falsepilot$ is not a phase shifted version of $\truepilot$.
Let 
\begin{equation}
D = \sum_{m=1}^M -\abs{x_m}^2 + \abs{y_m}^2 + Ax_m\conj y_m + Ay_m\conj x_m.
\end{equation}

From \cite[(B-9)]{ProakisDigComm} we have
\begin{align}
\Pr(D>0) &= 1+ \frac{(v_1v_2)^M}{2\pi \I}\int_{-\infty+\I\varepsilon}^{\infty+\I\varepsilon}\frac{dv}{v(v+\I v_1)^M(v-\I v_2)^M}, 
\end{align}
where 
\begin{align}
v_1 &= \sqrt{w^2 + \frac{1-\abs{\alpha}^2}{\rho\beta+1}} -w = \sqrt{w^2 - \frac{2w}{\rho\beta}} -w , \\
v_2 &= \sqrt{w^2 + \frac{1-\abs{\alpha}^2}{\rho\beta+1}} +w = \sqrt{w^2 - \frac{2w}{\rho\beta}} +w , \\
w &= \frac{-\rho\beta(1-\abs{\alpha}^2)}{2(\rho\beta+1)}.
\end{align}

We let 
\begin{align}
f(v)&=\frac{1}{v(v+\I v_1)^M(v-\I v_2)^M},
\\f_2(v) &= (v-\I v_2)^Mf(v) = \frac{1}{v(v+\I v_1)^M}.
\end{align}
Then, by the Cauchy residue theorem, we have
\begin{align}
\Pr(D>0)&= 1+-\frac{(v_1v_2)^M}{2\pi\I}\int_{-\infty+\I\varepsilon}^{\infty+\I\varepsilon}f(v)dv 
\\
&= 1+-\frac{(v_1v_2)^M}{2\pi\I} \left(2\pi\I \operatorname{Res}(f(v),\I v_2)\right) 
\\
&= 1+\frac{(v_1v_2)^M}{2\pi\I} \left(2\pi\I \frac{1}{(M-1)!}\lim_{v\to\I v_2} \frac{d^{M-1}}{dv^{M-1}}(v-\I v_2)^Mf(v)\right) 
\\ &= 1+\frac{(v_1v_2)^M}{2\pi\I} \left(2\pi\I \frac{1}{(M-1)!}\lim_{v\to\I v_2} \frac{d^{M-1}}{dv^{M-1}}f_2(v)\right) 
\\ &= 1+\frac{(v_1v_2)^M}{2\pi\I} \left(2\pi\I \frac{1}{(M-1)!} (-1)\sum_{k=0}^{M-1} \frac{(M+k-1)!}{k!}(v_2)^{-2M}\left(\frac{v_2}{v_1+v_2}\right)^{M+k}\right) 
\\ &= 1-\left(\frac{v_1}{v_2}\right)^M\sum_{k=0}^{M-1} {M+k-1\choose k}\left(\frac{v_2}{v_1+v_2}\right)^{M+k}
%
\\ 
& = 1- \left(\frac{v_1}{v_1+v_2}\right)^M \sum_{k=0}^{M-1}{M+k-1\choose k}\left(\frac{v_2}{v_1+v_2}\right)^k.
\end{align}

Now, introduce 
\begin{equation}
u = \frac{v_1}{v_1+v_2} = \frac{1}{2} - \frac{w}{2\sqrt{w^2-\frac{2w}{\rho\beta}}}.
\end{equation} 
Note that, $ \frac{1}{2} < u <  \frac{1}{2} - \frac{w}{2\abs{w}} = 1$, since $w<0$.
Then, 
\begin{align}
\Pr(D>0)  
&= 1 - u^M\sum_{k=0}^{M-1}\left(\!\!{M\choose k}\!\!\right)(1-u)^k \\
&= 1 - u^M\left(\underbrace{\sum_{k=0}^{\infty}\left(\!\!{M\choose k}\!\!\right)(1-u)^k}_{=(1-(1-u))^{-M}=u^{-M}}-\sum_{k=M}^{\infty}\left(\!\!{M\choose k}\!\!\right)(1-u)^k\right) \\
&\overset{(a)}{=}1-u^Mu^{-M}+u^M\sum_{k=M}^{\infty}\left(\!\!{M\choose k}\!\!\right)(1-u)^k\\
&=  u^M\sum_{k=M}^{\infty}\left(\!\!{M\choose k}\!\!\right)(1-u)^k \\
&=  \sum_{k=M}^{\infty}\left(\!\!{M\choose k}\!\!\right)u^M(1-u)^k, \label{eq:pep2}
\end{align}
where $\left(\!\!{n\choose k}\!\!\right) = {n+k-1\choose k} = \frac{(n+k-1)!}{(n-1)!k!}$ is the multiset coefficient,  and in $(a)$ we used that \linebreak $\frac{1}{(1-x)^n}~=~\sum_{k=0}^\infty \left(\!\!{n\choose k}\!\!\right)x^k$, when $\abs{x}<1$, which is a special case of the 	binomial series. 

Consider the logarithm of the $k$:th term in \eqref{eq:pep2}, 
\begin{equation}
\log\left(u^M(1-u)^k\right) = M\log u + k\log(1-u).
\end{equation}
The derivative with respect to $u$, 
\begin{equation}
\frac{d}{du}\left(  M\log u + k\log(1-u) \right) = \frac{M}{u} - \frac{k}{1-u} \leq \frac{M}{u} - \frac{M}{1-u} < 0, 
\end{equation}
is negative, since $k\geq M$ and $u > \frac{1}{2}$, and hence each of the terms in \eqref{eq:pep2} decrease with $u$, i.e., $\frac{d\Pr(D>0)}{du}<0$.
%
%
%
%
%
Further, we determine 
\begin{align}
\frac{du}{d\abs{\alpha}} &= \frac{du}{dw}\frac{dw}{d\abs{\alpha}} \\
&= \frac{ 1 }{2\left(w-\frac{2}{\rho\beta}\right)\sqrt{w^2-\frac{2w}{\rho\beta}}} \frac{\abs{\alpha}}{\rho\beta+1} \leq 0
\end{align}
and hence,
\begin{align}
\frac{\Pr(D>0)}{d\abs{\alpha}} &= \frac{\Pr(D>0)}{du}\frac{du}{d\abs{\alpha}} \geq 0.
\end{align}
Thus, we can conclude that the pairwise error probability increases when $\abs{\alpha}$ increases, which we intended to show.

\end{document}